%% file: ChoicelessLP.tex
\documentclass[a4paper]{article}
\usepackage{amssymb}
\usepackage{amsmath}
\usepackage{amsthm}
\usepackage[plainpages=false,pdfpagelabels]{hyperref}
\usepackage{xspace}
\usepackage{graphicx}
\usepackage{mdwlist}
\usepackage{color}
\usepackage{soul}
\usepackage{authblk}
\usepackage{algpseudocode}
\usepackage{xparse}
\usepackage{tikz}
\usepackage{3dplot}
\usepackage{changebar}
\let\originalparagraph\paragraph
\renewcommand{\paragraph}[2][.]{\originalparagraph{#2#1}}

\newtheorem{thm}{Theorem}
\newtheorem{lem}[thm]{Lemma}

\newtheorem{obs}[thm]{Observation}
\newtheorem{prop}[thm]{Proposition}
\newtheorem{definition}[thm]{Definition}
\newtheorem{problem}[thm]{Problem}

\newcounter{example}[section]
  {\hspace*{\fill}$\blacksquare$ \\}

\input{macros}

\title{Maximum Matching and Linear Programming \\in Fixed-Point Logic
with Counting\thanks{Research supported by EPSRC grant EP/H026835. An extended abstract of this paper will appear in the proceedings of LICS 2013.}}

\author{Matthew Anderson}
\author{Anuj Dawar}
\author{Bjarki Holm}
\affil{University of Cambridge Computer Laboratory\\
\texttt{firstname.lastname@cl.cam.ac.uk}}

\begin{document}
\maketitle

\begin{abstract}
\noindent
We establish the expressibility in fixed-point logic with counting 
(\FPC) of a number of natural polynomial-time problems. In particular, 
we show that the size of a maximum matching in a graph is definable in 
\FPC. This settles an open problem first posed by Blass, Gurevich and  
Shelah~\cite{BGS99}, who asked whether the existence of perfect matchings in
general graphs could be determined in the more powerful formalism of choiceless
polynomial time with counting. Our result is established by showing that the
ellipsoid method for solving linear programs can be implemented in \FPC. This
allows us to prove that linear programs can be optimised in \FPC if the
corresponding separation oracle problem can be defined in \FPC. On the way to
defining a suitable separation oracle for the maximum matching problem, we
provide \FPC formulas defining maximum flows and canonical minimum cuts in
capacitated graphs.
\end{abstract}

\pagebreak

\tableofcontents

\pagebreak

\section{Introduction}

\input{intro}

\section{Background}\label{sec:background}
\input{background}

\section{Expressing the Separation Problem in \FPC}\label{sec:sep}
\input{sep}

\section{Reducing Optimisation to Separation in \FPC}\label{sec:opt-sep}

\input{opt-sep}

\section{Application: Maximum Flow}\label{sec:maxflow}

\input{flows}

\section{Application: Minimum Cut}\label{sec:mincut}

\input{cuts}

\section{Application: Minimum Odd Cut}\label{sec:minoddcut}

\input{min-odd-cuts}

\section{Application: Maximum Matching}\label{sec:matching}

\input{matching}

\section{Conclusion}

\input{conc}

\section*{Acknowledgments}

The authors would like to thank Siddharth Barman for his helpful comments on an
early draft of this paper and the anonymous reviewers for their constructive
suggestions.
%%  and
%% kind comments.

%\section{Biblography}
\bibliographystyle{amsalpha}
\bibliography{ChoicelessLP}

%\appendices
% \input{ellipsoid_long}
%\input{ellipsoid_fulldim_appendix}
%\input{appendix}

\end{document}

%% file: macros.tex
\newcommand{\de}{:=}
\newcommand{\set}[1]{\ensuremath{\{{#1}\}}}
\newcommand{\condset}[2]{\set{{#1} \; | \; {#2}}}
\newcommand{\bigset}[1]{\ensuremath{\left\{{#1}\right\}}}
\newcommand{\bigcondset}[2]{\bigset{\left.{#1} \; \right| \; {#2}}}

\newcommand{\NN}{\ensuremath{\mathbb{N}}\xspace}
\newcommand{\ZZ}{\ensuremath{\mathbb{Z}}\xspace}

\newcommand{\QQ}{\ensuremath{\mathbb{Q}}\xspace}

\newcommand{\nnQQ}{\ensuremath{\mathbb{Q}_{\ge 0}}\xspace}

\newcommand{\tup}[1]{\vec{#1}}
\newcommand{\logic}[1]{\text{\upshape #1}\xspace}
\newcommand{\FPC}{\logic{FPC}}
\newcommand{\FPR}{\logic{FPR}}
\newcommand{\FOC}{\logic{FOC}}

\newcommand{\CPTC}{\logic{$\tilde{\text C}$PT(Card)}}

\newcommand{\countingTerm}[1]{\#_{#1}}

\newcommand{\vocvec}{\tau_{\text{vec}}}

\newcommand{\vocmat}{\tau_{\text{mat}}}
\newcommand{\vocint}{\tau_{\ZZ}}
\newcommand{\vocnum}{\tau_{\QQ}}
\newcommand{\vocmatch}{\tau_{\text{match}}}

\newcommand{\logicoperator}[1]{\mathbf{#1}}
\newcommand{\ifpop}{\logicoperator{ifp}}
\DeclareMathOperator{\fin}{fin}

\newcommand{\struct}[1]{\ensuremath{\mathbf #1}\xspace}

\newcommand{\univ}[1]{\ensuremath{\dom(\struct #1)}} 
\newcommand{\isom}{\cong}

\newcommand{\poly}[2]{P_{{#1},{#2}}}

\newcommand{\norm}[1]{|\!| #1 |\!|}

\newcommand{\inorm}[1]{\norm{#1}_{\infty}}

\newcommand{\PT}{\ensuremath{\mathrm{P}}\xspace}
\newcommand{\NP}{\ensuremath{\mathrm{NP}}\xspace}

\newcommand{\Pp}{\ensuremath{\mathcal{P}}\xspace}

\newcommand{\fold}[2]{\ensuremath{{[{#1}]^{#2}}}}
\newcommand{\afold}[2]{\fold{#1}{\tilde{#2}}}
\newcommand{\bigfold}[2]{\ensuremath{{\left[{#1}\right]^{#2}}}}

\DeclareDocumentCommand\fd{m g}{\fold{#1}{\IfNoValueTF{#2}{\sigma}{#2}}}
\DeclareDocumentCommand\afd{m g}{\afold{#1}{\IfNoValueTF{#2}{\sigma}{#2}}}
\DeclareDocumentCommand\unfd{m g}{\fold{#1}{\IfNoValueTF{#2}{-\sigma}{#2}}}
\DeclareDocumentCommand\bigfd{m g}{\bigfold{#1}{\IfNoValueTF{#2}{\sigma}{#2}}}

\newcommand{\conv}{\mathrm{conv}}
\newcommand{\cone}{\mathrm{cone}}

\newcommand{\order}[1]{\textsf{#1}}

\newcommand{\val}{\ensuremath{\mathrm{val}}}
\newcommand{\bs}{\ensuremath{\backslash}}
\newcommand{\fl}{$(s,t)$-flow\xspace}
\newcommand{\fls}{$(s,t)$-flows\xspace}
\newcommand{\cu}{$(s,t)$-cut\xspace}
\newcommand{\cus}{$(s,t)$-cuts\xspace}
\newcommand{\MC}[1]{{K_{#1}}}

\newcommand{\Coll}{\textsc{Collapse}\xspace}
\newcommand{\WMOC}{\textsc{WitMinOddCut}\xspace}
\newcommand{\FPCOpt}{{\text{\sc Opt}$^*$}\xspace}
\newcommand{\Opt}{{\text{\sc Opt}}}

\newcommand{\Refine}{{\text{\sc Refine}}\xspace}

\newcommand{\sse}{\subseteq}
\newcommand{\ssn}{\subsetneq}
\newcommand{\spe}{\supseteq}

\newcommand{\es}{\emptyset}

\renewcommand{\bar}[1]{\overline{#1}}
\newcommand{\dom}{\mathrm{dom}}
\newcommand{\sep}{\;|\;}
\newcommand{\bit}{\operatorname{bit}}

\newcommand{\defeq}{:=}

\newcommand{\alignedeq}[1]{
\begin{equation*}
\begin{aligned}
#1
\end{aligned}
\end{equation*}
}

\newcommand{\MEndIf}{\vspace{-.6ex}\EndIf}

\newcommand{\algxsetup}{
  \algdef{Se}[BLOCK]{Block}{EndBlock}[0]{}{}%
  \algtext*{Block}%
  \algtext*{EndBlock}%
  \algtext*{EndIf}%
  \algtext*{EndWhile}%
  \algtext*{EndFor}%
  \algrenewcommand\algorithmicindent{1.25em}%
  \algrenewcommand{\algorithmiccomment}[1]{ \hfill// ##1}%
  \algrenewcommand{\algorithmicprocedure}{\textbf{oracle}}%

}

\newcommand{\algxio}[7]{
  \algxsetup
  \begin{figure}
  \caption{#3}
  \label{#2}
 
  \smallskip
  \hrule
  \smallskip
  {\sc #1}$(#4)$
  \smallskip
  \hrule
  \smallskip
  \hspace*{0.75ex}Input: {#5} %\\
  \hspace*{0.75ex}Output: {#6} \strut
  \hrule
  \smallskip
  \begin{algorithmic}[1]
  #7%
  \end{algorithmic}
  \smallskip
  \hrule
  \end{figure}
}

\NewDocumentCommand{\WSEP}{g}{\ensuremath{\mathrm{WSEP}\IfNoValueTF{#1}{}{({#1})}}\xspace}
\NewDocumentCommand{\SEP}{g}{\ensuremath{\mathrm{SEP}\IfNoValueTF{#1}{}{({#1})}}\xspace}
\NewDocumentCommand{\WOPT}{g}{\ensuremath{\mathrm{WOPT}\IfNoValueTF{#1}{}{({#1})}}\xspace}
\NewDocumentCommand{\OPT}{g}{\ensuremath{\mathrm{OPT}\IfNoValueTF{#1}{}{({#1})}}\xspace}
\NewDocumentCommand{\BOUND}{g}{\ensuremath{\mathrm{BOUND}\IfNoValueTF{#1}{}{({#1})}}\xspace}
\NewDocumentCommand{\CIRCUM}{g}{\ensuremath{\mathrm{CIRCUM}\IfNoValueTF{#1}{}{({#1})}}\xspace}

\newcommand{\SIZE}[1]{\size{#1}}
\newcommand{\size}[1]{\ensuremath{\langle{#1}\rangle}}
\newcommand{\fSIZE}[1]{\size{#1}_{\mathrm{f}}}
\newcommand{\vSIZE}[1]{\size{#1}_{\mathrm{v}}}

\newcommand{\case}[2]{\smallskip \noindent{#1} \emph{#2} \smallskip}

%% file: intro.tex
The question of whether there is a logical characterisation of
the class P of problems solvable in polynomial time,
first posed by Chandra and Harel~\cite{CH82}, has been a central research question in
descriptive complexity for three decades. At one time it was
conjectured that \FPC, the extension of inflationary fixed-point logic
by counting terms, would suffice to express all polynomial-time
properties, but this was refuted by Cai, F\"urer and Immerman~\cite{CFI92}.
Since then, a number of logics have been
proposed whose expressive power is strictly greater than that
of \FPC but still contained within \PT. Among these are \FPR,
fixed-point logic with rank operators~\cite{DGHL09}, and \CPTC,
choiceless polynomial time with counting~\cite{BGS99,BGS02}.  For both
of these it remains open whether their expressive
power is strictly weaker than \PT. 

Although it is known that \FPC 
does not express all polynomial-time computable properties, 
% expresses strictly less than \PT, t
the descriptive power of \FPC
still forms a natural class within \PT.  For instance, it has been shown that \FPC can express
all polynomial-time properties on many natural graph classes, such as any class
of proper minor-closed graphs~\cite{Gro10}. Delimiting the expressive power
of \FPC therefore remains an interesting challenge. In particular, it is of interest to
establish what non-trivial polynomial-time algorithmic techniques can be
expressed in this logic.  The conjecture that \FPC captures \PT was based on the
intuition that the logic can define all ``obvious'' polynomial-time algorithms.  
The result of Cai et al.\ and the subsequent work of Atserias et
al.~\cite{ABD09} showed that one important technique---that of Gaussian
elimination for matrices over finite fields---is not captured by \FPC.
The question remains what other natural problems for which membership in $\PT$
is established by non-trivial algorithmic methods might be expressible in \FPC.

For instance, it was shown by Blass et al.~\cite{BGS02} that there is a sentence of \FPC
that is true in a bipartite graph $G$ if, and only if, $G$ contains a perfect
matching.  They posed as an open question whether the existence of a perfect
matching on general graphs can be defined in \CPTC (see also~\cite{Bla05,Ros10}
for more on this open question).  Indeed, this question first appears
in~\cite{BGS99} where it is stated that it seems ``unlikely'' that this problem
can be decided in \CPTC.  One of our main contributions in this paper is to settle this question by showing that the size of a
maximum matching in a general graph can be defined in \FPC (and therefore also
in \CPTC).

%\medskip

On the way to establishing this result, we show that a number of other
interesting algorithmic problems can also be defined in \FPC.  
To begin with, we study the logical definability of \emph{linear programming} problems.
%More specifically, 
Here we show that there is a formula of \FPC which defines on a polytope 
a point inside the polytope which maximises a given linear objective function, if such a point 
exists. Here, by a polytope we mean a convex set in Euclidean space given by finite intersections of linear
inequalities (or \emph{constraints}) over a set of variables, suitably represented as a relational
structure without an ordering on the sets of variables or constraints.

% \marginnote{M: Is ``\FPC formula'' appropriate here and in the last paragraph of this section?}

More specifically, we consider representations where, as in many applications of linear programming, 
the set of constraints is
not given explicitly (indeed, it may be exponentially large) but is
determined instead by a \emph{separation oracle}.  This is a procedure
which, given a candidate point $x$, determines whether $x$ is feasible
and, if it is not, returns a constraint that is violated by $x$.  It is well known that Khachiyan's polynomial-time algorithm for linear programming---the \emph{ellipsoid method}---can be extended to prove that the linear programming and separation problems are
polynomial-time equivalent (c.f., \cite{K80,GLS81,GLS88}).  We show
an analogous result
% statement 
for \FPC: if a separation oracle for a polytope is expressible
in \FPC, then linear programming on that polytope is definable
in \FPC.  Informally speaking, the idea is the following.  Although the set of
variables % in the constraints are 
is not inherently ordered, the
separation oracle induces a natural equivalence relation on these
variables whereby two variables are equivalent if they cannot be
distinguished in any invocation of the oracle.
% (this is made more precise in Section~\ref{sec:opt-sep}).   
Given a linear ordering on these
equivalence classes, we can define in \FPC a reduction of the optimisation 
problem to an instance with an ordered set of variables by
taking the quotient of the polytope under the induced equivalence relation.
We show
that solving the optimisation problem on this quotiented polytope---now using the
classical polynomial-time reduction made possible by the ordering---allows us to recover a
solution to the original problem. In practice, neither the equivalence classes
nor the order are given beforehand; rather they are iteratively refined via the
invocations of the separation oracle made while optimising over the quotiented
polytope. 
%% the steps of constructing the equivalence relation and solving the
%% optimisation are interleaved.  We describe a procedure which %where we
%% maintains a pre-order on the variables and if the reduction to
%% corresponding quotient polytope fails, it is due to an invocation of
%% the separation oracle that forces a further refinement of the
%% equivalence relation.  
The details of this result are presented in
Section~\ref{sec:opt-sep}.

Thus to express a problem modelled by a linear program in \FPC it suffices to express a separation
oracle in \FPC.  A key difficulty to expressing separation oracles in \FPC is
that a particular violated constraint must be chosen.  We show in
Section~\ref{sec:sep} that when the constraints are given explicitly,
a \emph{canonical} violated constraint can be defined by taking the sum of all the violated
constraints.  This implies that the class of feasible linear programs,
when explicitly given, can be
expressed in \FPC.  When the constraints are not given explicitly, it may still
be possible to express canonical violated constraints (and hence separation
oracles) by using domain knowledge.

As a first application of the \FPC-definability of explicitly-given linear
programs, we show that a maximum flow in a capacitated graph is definable
in \FPC.  Indeed, this follows rather directly from the first result, since the
flow polytope is of size polynomial in $G$ and explicitly given, and hence a
separation oracle can be easily defined from $G$ in \FPC.  These results are presented
in Section~\ref{sec:maxflow}.

Next, we use the definability of maximum flows to show that minimum cuts are
also definable in \FPC. That is, in the vocabulary of capacitated graphs, there
is a formula which defines a set of vertices $C$ corresponding to a minimum
value cut separating $s$ from $t$.  The cut $C$ defined in this way is canonical in a strong
sense, in that we show that it is the \emph{smallest} (under set-inclusion)
minimum cut separating $s$ from $t$. The definition of minimum cut and useful
variants are presented in Section~\ref{sec:mincut}.

Finally, we turn to the maximum matching problem.  For a graph $G=(V,E)$, the
matching polytope is given by a set of constraints of size exponential in the
size of $G$.  We show that there is a separation oracle for this set definable
from $G$ in \FPC, using the definability of minimum cuts.  To be precise, we use
the fact that a separation oracle for the matching polytope can be
obtained from a computation of minimum odd-size cuts in a graph~\cite{PR82}.  
In Section~\ref{sec:minoddcut} we prove that
there is always a pair of vertices $s,t$ such that a canonical minimum \cu is a
minimum odd-size cut.  This, combined with the definability of canonical minimum
cuts, gives us the separation oracle for matching that we seek. Note that it is not
possible in general to actually define a canonical maximum matching. To see this, consider $K_n$, the complete
graph on $n$ vertices. This graph contains an exponential number of maximum matchings
and for any two of these matchings, there is an automorphism of the graph taking one to the
other. Thus, it is not possible for any formula of \FPC (which is necessarily
invariant under isomorphisms) to pick out a particular matching.  What we can do, however,
is to define a formula that gives the size of the maximum matching in a graph. This, in turn, 
enables us to write a sentence of \FPC that is true in a graph $G$ if, and only
if, it contains a perfect matching. Our results on matchings are presented in
Section~\ref{sec:matching}.

%% file: background.tex
We write $[n]$ to denote the set of positive integers $\{0,\ldots,n-1\}$.  Given
sets $I$ and $A$, a column vector $u$ over $A$ indexed by $I$ is a function $u:
I \rightarrow A$, and we write $A^I$ for the set of all such vectors.
Similarly, an $I,J$-matrix over $A$ is a function $M: I\times J \rightarrow A$
and we write $M_{ij}$ for $M(i,j)$ and $M_i$ for the row (vector) of
$M$ indexed by $i$.  For an integer $z$, $|z|$ denotes its absolute value.  For
a vector $v \in \QQ^I$, %$\|v\|$ denotes its Euclidean norm and
$\|v\|_\infty \de \max_{i \in I} |v_i|$ denotes its infinity norm.

% -----------------------------------------------------------------------------
%
% Logics and structures
%
% -----------------------------------------------------------------------------
\subsection{Logics and Structures}

A relational \emph{vocabulary} $\tau$ is a finite sequence of relation and
constant symbols $(R_1, \dots, R_k, c_1, \dots, c_\ell)$, where every relation
symbol $R_i$ has a fixed \emph{arity} $a_i \in \NN$. A structure $\struct A =
(\univ A, R_1^{\struct A}, \dots, R_k^{\struct A}, c_1^{\struct A}, \dots,
c_\ell^{\struct A})$ over the vocabulary $\tau$ (or a \emph{$\tau$-structure})
consists of a non-empty set $\univ A$, called the \emph{universe} of $\struct
A$, together with relations $R_i^{\struct A} \subseteq \univ A^{a_i}$ and
constants $c_j^{\struct A} \in \univ A$ for each $1 \leq i \leq k$ and $1 \leq
j \leq \ell$. Members of the set $\univ A$ are called the \emph{elements} of
$\struct A$ and we define the \emph{size} of $\struct A$ to be the cardinality
of its universe.  In what follows, we often consider multi-sorted structures.
That is, $\univ A$ is given as the disjoint union of a number of
different \emph{sorts}.  In this paper we consider only finite structures, that
is structures over a finite universe.  For a particular vocabulary $\tau$ we use
$\fin[\tau]$ to denote the set of all finite $\tau$-structures.

% -----------------------------------------------------------------------------
% FPC
% -----------------------------------------------------------------------------
\paragraph{Fixed-point logic with counting} 

Fixed-point logic with counting (\FPC) is an extension of inflationary
fixed-point logic with the ability to express the cardinality of definable
sets. The logic has two types of first-order variable: \emph{element variables},
which range over elements of the structure on which a formula is interpreted in
the usual way, and \emph{number variables}, which range over some initial
segment of the natural numbers. We traditionally write element variables with
lower-case Latin letters $x, y, \dots$ and use lower-case Greek letters
$\mu, \eta, \dots$ to denote number variables.
% \marginnote{M: Is this remark about Latin and Greek variable labelling necessary? B: Well, we do use this in the next paragraph. No strong feelings about this though.}

The atomic formulas of $\FPC[\tau]$ are all formulas of the form: $\mu = \eta$ or
$\mu \le \eta$, where $\mu, \eta$ are number variables; $s = t$ where $s,t$ are
element variables or constant symbols from $\tau$; and $R(t_1, \dots, t_m)$,
where each $t_i$ is either an element variable or a constant symbol and $R$ is a
relation symbol of
arity $m$. The set $\FPC[\tau]$ of \emph{$\FPC$
formulas} over $\tau$ is built up from the atomic formulas by applying an
inflationary fixed-point operator $[\ifpop_{R,\tup x}\phi](\tup t) $;
forming \emph{counting terms} $\countingTerm{x} \phi$, where $\phi$ is a formula
and $x$ an element variable; forming formulas of the kind $s = t$ and $s \le t$
where $s,t$ are number variables or counting terms; as well as the standard
first-order operations of negation, conjunction, disjunction, universal and
existential quantification. Collectively, we refer to element variables and
constant symbols as \emph{element terms}, and to number variables and counting
terms as \emph{number terms}.

%% To define the semantics of $\FPC[\tau]$, we consider \emph{numerical
%% $\tau$-structures}, which are $\tau$-structures extended with an
%% initial segment of the integers. More specifically, let
%% $\tau^* \defeq \tau \cup \{ \le \}$, where we assume that $\tau$ does
%% not contain the relation symbol $\le$, and consider a finite
%% $\tau$-structure $\struct A$ over a universe $A$ of size $n$. Then we
%% write $\numstruct A$ to denote the $\tau^*$-structure whose universe
%% is $A \disjoint [n]$, where all relations and constant symbols in
%% $\tau$ are interpreted as in $\struct A$ and where $\leq$ is
%% interpreted as the natural ordering over $[n]$. We view $\numstruct A$
%% as a two-sorted structure, with \emph{element sort} $A$
%% and \emph{number sort} $[n]$. 
%% Formally, formulas of $\FPC[\tau]$ are interpreted over pairs
%% $(\numstruct A, \alpha)$, where $\numstruct A$ is a numerical
%% $\tau$-structure and $\alpha$ a variable assignment in $\numstruct
%% A$. If $t$ is a term, then we write $\alpha(t)$ to denote the value
%% that is assigned to $t$ over $\numstruct A$. 

For the semantics, number terms take values in $[n+1]$ and element terms take
values in $\dom(\struct A)$ where $n \de |\dom(\struct A)|$. The semantics of
atomic formulas, fixed-points and 
first-order operations are defined as usual (c.f., e.g., \cite{Ebb99} for
details), with comparison of number terms $\mu \le \eta$ interpreted by
comparing the corresponding integers in $[n+1]$. Finally, consider a counting
term of the form $\countingTerm{x}\phi$, where $\phi$ is a formula and $x$ an
element variable. Here the intended semantics is that $\countingTerm{x}\phi$
denotes the number (i.e., the element of $[n+1]$) of elements that satisfy the
formula $\phi$.
%%  More formally, the semantics of counting terms is defined as follows:
%% %
%% \[
%% 	\alpha(\countingTerm{x}\phi) \defeq 
%% 	\card{\{ a \in A \sep \struct A \models \phi[\alpha\tfrac{a}{x}] \}},
%% \]

%% \noindent
%% where $\alpha\tfrac{a}{x}$ is the assignment obtained from $\alpha$ by
%% assigning the element $a \in A$ to the variable $x$. 
In general, a formula $\phi(\tup x,\tup \mu)$ of \FPC defines a relation over
$\univ A \uplus [n+1]$ that is invariant under automorphisms of $\struct A$.
For a more detailed definition of $\FPC$, we refer the reader to~\cite{Ebb99,
Libkin}.

It is known, by the results of Immerman and Vardi~\cite{I86,V82}, that
every polynomial-time decidable property of \emph{ordered} structures is
definable in fixed-point logic, and therefore also in \FPC.  Here, an
ordered structure is one which includes a binary relation which is a
linear order of its universe.
Throughout this paper, we refer to this result as the \emph{Immerman-Vardi theorem}.

% -----------------------------------------------------------------------------
% Logical interpretations
% -----------------------------------------------------------------------------
\paragraph{Logical interpretations}

We frequently consider ways of defining one structure within another in some
logic $\logic L$, such as first-order logic or fixed-point logic with
counting. Consider two vocabularies $\sigma$ and $\tau$ and a logic $\logic
L$. An \emph{$m$-ary $\logic L$-interpretation of $\tau$ in $\sigma$} is a
sequence of formulae of $\logic L$ in vocabulary $\sigma$ consisting of: (i) a
formula $\delta(\tup x)$; (ii) a formula $\varepsilon(\tup x, \tup y)$; (iii)
for each relation symbol $R \in \tau$ of arity $k$, a formula $\phi_R(\tup
x_1, \dots, \tup x_k)$; and (iv) for each constant symbol $c \in \tau$, a
formula $\gamma_c(\tup x)$, where each $\tup x$, $\tup y$ or $\tup x_i$ is an
$m$-tuple of free variables. We call $m$ the \emph{width} of the
interpretation. We say that an interpretation $\Theta$ associates a
$\tau$-structure $\struct B$ to a $\sigma$-structure $\struct A$ if there is a
surjective map $h$ from the $m$-tuples $\{ \tup a \in (\univ{A} \uplus
[n+1])^m \sep \struct A \models \delta[\tup a] \}$ to $\struct B$ such that:

\begin{itemize}
\item $h(\tup a_1) = h(\tup a_2)$ if, and only if, $\struct
A \models \varepsilon[\tup a_1, \tup a_2]$; 
	
\item $R^\struct{B}(h(\tup a_1), \dots, h(\tup a_k))$ if, and only if, $\struct
A \models \phi_R[\tup a_1, \dots, \tup a_k]$; 
	
\item $h(\tup a) = c^\struct{B}$ if, and only if, $\struct
A \models \gamma_c[\tup a]$.

\end{itemize}

\noindent 
Note that an interpretation $\Theta$ associates a $\tau$-structure with $\struct
A$ only if $\varepsilon$ defines an equivalence relation on $(\univ{A} \uplus
[n+1])^m$ which is a congruence with respect to the relations defined by the
formulae $\phi_R$ and $\gamma_c$. In such cases, however, $\struct B$ is
uniquely defined up to isomorphism and we write $\Theta(\struct
A) \defeq \struct B$.

It is not difficult to show that formulas of \FPC compose with reductions in the
sense that, given an interpretation $\Theta$ of $\sigma$ in $\tau$ and a
$\sigma$-formula $\phi$, we can define a $\tau$-formula $\phi'$ such that
$\struct A \models \phi'$ if, and only if, $\Theta(\struct A) \models \phi$
(see \cite[Sec.~3.2]{Imm99})  %XXX - \sigma and \tau were switched.
In particular, if $\Theta(\struct A)$ is an ordered structure, for all
$\struct A$, then by the Immerman-Vardi theorem above, for any
polynomial-time decidable class $C$, there is an \FPC formula
$\phi$ such that $\struct A \models \phi$ if, and only if,
$\Theta(\struct A) \in C$.

% -----------------------------------------------------------------------------
%
% Matrices and vectors indexed by unordered sets
%
% -----------------------------------------------------------------------------
\subsection{Numbers, Vectors and Matrices}

Let $z$ be an integer, $b \geq \lceil\log_2(|z|)\rceil$, $B = [b]$ and
write $\bit(x, k)$ to denote the $k$-th least-significant bit in the binary
expansion of $x \in \NN$.  We view the integer $z = s \cdot x$ as a product of a sign
$s \in \set{-1,1}$ and a natural number $x$.   We can represent $z$ as a
single-sorted structure $\struct B$ on a domain of bits $B$ over the vocabulary
$\vocint \de \set{X, S, \leq_B}$. Here $\leq_B$ is interpreted as a linear
ordering of $B$, the unary relation $S$ indicates that the sign $s$ of the
integer is 1 if $S^{\struct B} = \es$ and $-1$ otherwise, and the unary relation
$X$ is interpreted as $X^{\struct B} = \condset{k \in B}{\bit(x,k) = 1}.$ That
is $k \in X^{\struct B}$ when the ``the $k$-th bit in the binary expansion of
$x$ is 1.''  Similarly we consider a \emph{rational number} $q =
s \cdot \frac{x}{d}$ as a structure on the domain of bits $B$ over
$\vocnum \de \set{ X, D, S, \leq_B}$, where $X$ and $S$ are as before and $D$ is
interpreted as the binary encoding of the denominator $d$ when $D^{\struct
B} \neq \es$.  

We now generalise these notions and consider unordered tensors over the
rationals (the case of integers is completely analogous).  Let $J_1, \ldots, J_r$ be
a family of finite non-empty sets.  An unordered tensor $T$ over $\QQ$ is a
function $T : J_1 \times \cdots \times J_r \rightarrow \QQ$.  We write
$t_{j_1\ldots j_r} = s_{j_1\ldots j_r} \frac{x_{j_1\ldots j_r}}{d_{j_1\ldots
j_r}}$ to denote the element of $T$ indexed by $(j_1,\ldots,j_r) \in
J_1 \times \cdots \times J_r$.  
Writing $m \in \NN$ for the the maximum absolute value of integers appearing as 
either numerators or denominators of elements in the range of $T$, let
$b \geq \lceil\log_2(|m|)\rceil$ and $B = [b]$.
The tensor $T$ is then an $(r+1)$-sorted structure
$\struct T$ with $r$ index sorts $J_1,\ldots,J_r$ and a bit sort $B$ over the
vocabulary $\sigma_{\text{ten},r} \de \set{X,D,S,\leq_B}$.  Here $\leq_B$ is
interpreted as before, the $(r+1)$-ary relation $S$ is interpreted as indicating
the value of the sign $s_{j_1\ldots j_r} \in \set{-1,1}$ as before, the $(r+1)$-ary
relation $X$ is interpreted as $$
%N^{\struct T} = 
\condset{(j_1,\ldots,j_r,k) \in
J_1 \times \cdots \times J_r \times B}{\bit(x_{j_1\ldots j_r},k) = 1},$$ and the
$(r+1)$-ary relation $D$ is similarly interpreted as the binary representation
of the denominators of $T$.  We are only interested in the case of rational
vectors and matrices and so define the vocabularies
$\vocvec \de \sigma_{\text{ten},1}$ and $\vocmat \de \sigma_{\text{ten},2}$.

In~\cite{H10} it is shown that a variety of basic linear-algebraic operations on
rational vectors and matrices described in this way can be expressed in
fixed-point logic with counting. These include computing equality, norms, dot
product, matrix product, determinant and inverse.
% and rank.  
%% Frequently operations
%% on matrices and vectors require dimensions indexed in the same way.  As
%% convention we extend the index sets of the objects to the union of the objects'
%% index sets.  In doing so we view the mathematical objects as residing in a
%% larger space where unspecified elements take the value 0.  For example, consider
%% the dot product of vectors $u \in \QQ^I$ and $v \in \QQ^J$: We write $u^\top v
%% = \sum_{k \in I \cup J} u_k v_k$, where $u_k = 0$ for $k \in J \bs I$ and $v_k =
%% 0$ for $k \in I \bs J$.

% -----------------------------------------------------------------------------
%
% Linear programming
%
% -----------------------------------------------------------------------------
\subsection{Linear Programming}

We recall some basic definitions from 
%polyhedral 
combinatorics and linear optimisation. For further background, see, for example, the textbook by Gr\"{o}tschel et al.~\cite{GLS88}. % or Korte and Vygen~\cite{KV12}.

% -----------------------------------------------------------------------------
% Polytopes, polyhedra and constraint matrices
% -----------------------------------------------------------------------------
%% \subsubsection{Geometry}

%% Consider the real Euclidean space $\RR^V$ indexed by a set $V$.  

%% For a point
%% $x \in \RR^V$ and $\epsilon \in \nnRR$ let $\sphere{x}{\epsilon} \de \condset{
%% y \in \RR^V}{\norm {x - y} \le \epsilon}$ denote the \emph{sphere} of radius
%% $\epsilon$ about $x$. This notion intuitively extends to
%% $\sphere{K}{\epsilon} \defeq \bigcup_{x \in K} \sphere x \epsilon$ for any set
%% $K \subseteq \RR^V$.
%% % and similarly let $\sphere{K}{-\epsilon}$ denote the
%% %elements of $K$ which are at distance at least $\epsilon$ from any point not in
%% %$K$.
%% A positive definite matrix $E \in \RR^{V\times V}$ and point
%% $x \in \RR^V$ specify an ellipsoid 
%% %$\Ell(E,x)$ where 
%% $$\Ell(E,x) \de \condset{y \in \RR^V}{(y-x)^\top E^{-1} (y-x) \le 1}.$$
%% Because $E$ is positive definite, $\Ell(E,x)$ is full dimensional in $\RR^V$ and
%% has non-zero volume $\Vol(\Ell(E,x))$.

\paragraph{Polytopes} 
Consider the rational Euclidean space $\QQ^V$ indexed by a set
$V$. The solutions to a system of linear equalities and   
inequalities over $\QQ^V$ is the intersection of some number
of \emph{half-spaces} of the kind $\condset{x \in \QQ^V}{a^\top x \leq b}$
specified by the \emph{constraint} $a^\top x \le b$, where $a \in \QQ^V$ and
$b \in \QQ$. A (rational) \emph{polytope} is a convex set $P \subseteq \QQ^V$ which is the
intersection of a \emph{finite} number of half-spaces. That is to say, there are
a set of constraints $C$, a \emph{constraint matrix} $A \in \QQ^{C \times V}$
and vector $b \in \QQ^C$, such that $P = \poly A
b \defeq \condset{x \in \QQ^V}{A x \leq b}$.  

Polytopes have an alternative characterisation as a combination of convex hulls and cones. 
Let $S$ be a finite set of points in $\QQ^V$ and define
the \emph{convex hull} of $S$ $$\conv(S) \de \bigcondset{\sum_{s \in
S} \lambda_s s}{\lambda_s \in \nnQQ, \forall s \in S\, \text{ and
} \sum_{s \in S} \lambda_s= 1},$$ and similarly define the \emph{cone} of
$S$ $$\cone(S) \de \bigcondset{\sum_{s \in S} \lambda_s s}{\lambda_s \in \nnQQ, \forall s \in S}.$$ 
If $P$ is a polytope in $\QQ^V$, then there exist finite
sets $S_1, S_2 \sse \QQ^V$ such that $P = P_{S_1,S_2} \de \conv(S_1)
+ \cone(S_2) = \condset{x_1 + x_2}{x_1 \in \conv(S_1), x_2 \in \cone(S_2)}$.

% \marginnote{BH: do we need a reference for the $P = \conv(C)
% + \cone(E)$ fact? MA: Well, its discussed in Chapter 0 of \cite{GLS88} without a
% formal statement or reference; as far as I can tell it is just a standard fact
% about polytopes.}   

%% We call a polytope $P$ \emph{rational} if there exists an $A \in \QQ^{C \times
%% V}$ and $b \in \QQ^C$ such that $P = \poly A b$ (or equivalently,
%% $S_1,S_2 \sse \QQ^V$ such that $P = \poly {S_1} {S_2}$).  

The \emph{size}, or bit complexity, of a vector $c \in \QQ^V$ (denoted
$\size{c}$) is the number of bits required to encode the components of $c$ in
some standard encoding of rational numbers. Note that $\size{c}$ is at least $|V|$.  The size of a
constraint $a^\top x \le b$ is then $\SIZE{\binom{a}{b}}$.  If $A x \le b$ is a
system of linear inequalities then its size is the maximum over the sizes of its
individual constraints.  Note that this measure is explicitly independent of the
number of constraints in the system.  The \emph{facet complexity} $\fSIZE{P}$ of
a polytope $P$ is the minimum over the sizes of the systems $Ax \le b$
such that $P = \poly A b$.  The \emph{vertex complexity} $\vSIZE{P}$ of a
polytope $P$ is the minimum over the maximum size vector in the union
of sets $S_1, S_2 \sse \QQ^V$ such that $P = \conv(S_1) + \cone(S_2)$.  The facet
and vertex complexity of a polytope are closely related: $\vSIZE{P} \le
4|V|^2\fSIZE{P}$ and $\fSIZE{P} \le 3|V|^2\vSIZE{P}$ \cite[Lemma 6.2.4]{GLS88}.

%% A polytope which is also compact is called a \emph{polyhedron}.\footnote{Be
%% aware that \cite{GLS88} uses the opposite notation and calls a convex set
%% satisfying $Ax \le b$ a polyhedron and a compact polyhedron a polytope.} 

%% A polyhedron $P$ is said to be \emph{$R$-circumscribed} when there is an
%% $R \in \nnRR$ such that $P \sse \sphere{0}{R}$.  Furthermore we say that $P$
%% has \emph{full dimension} if there is a point $x \in \RR^V$ and $r \in \pRR$
%% such that $\sphere{x}{r} \sse P$.  Note that if $P$ has full dimension then
%% it has non-zero volume.  A polytope is \emph{bounded} if it has full
%% dimension and is circumscribed.  More specifically, we say $P$ is
%% $(x,r,R)$-\emph{bounded} if $\sphere{x}{r} \sse P \sse \sphere{x}{R}.$

% -----------------------------------------------------------------------------
% Linear programming problems
% -----------------------------------------------------------------------------
\paragraph{Problems on polytopes}
We are interested in two main combinatorial problems on polytopes: \emph{linear optimisation} and \emph{separation}.

\begin{problem}[Linear Optimisation]
  Let $V$ be a set, $P \sse \QQ^V$ be a polytope and $c \in \QQ^V$.
The \emph{linear optimisation problem} 
%\OPT{P,c} 
on $P$
is the problem of determining either (i) an element $y \in P$ such that $c^\top y = \max\condset{c^\top
x}{x \in P}$, (ii) that $P = \es$ or (iii) that $P$ is unbounded in the
direction of $c$. % (i.e., $\max\condset{c^\top x}{x \in P}$ is unbounded).
\end{problem}

An instance of the linear optimisation problem is 
%often 
called a \emph{linear program} and the linear function $x \mapsto c^\top x$ is 
%often 
called the \emph{objective function}.
Over the years, a number of algorithms for solving linear programs have been studied.
Early work by Dantzig~\cite{D98} gave a combinatorial
algorithm---the \emph{simplex method}---which traverses the vertices (extremal points) of the polytope favouring
vertices that improve the objective value.  Although the simplex method is
useful in practice, it tends not to be theoretically useful because strong
worst-case performance guarantees are not known\footnote{Stronger guarantees are
known for the average-case and smoothed complexity of the simplex
method \cite{ST04}.}.  A series of works studying linear programming from a
geometric perspective~\cite{S72,Y76,S77} culminated in the breakthrough of
Khachiyan~\cite{K79,K80} which established a polynomial-time
algorithm---the \emph{ellipsoid method}---for solving linear programs.
One of the strengths of the ellipsoid method is that it 
can be applied to linear programs where the constraints are not given explicitly.
In such implicitly-defined linear programs, we are instead given a polynomial-time algorithm, 
known as a \emph{separation oracle}, 
for solving the following ``separation problem''.

\begin{problem}[Separation]
  Let $V$ be a set, $P \sse \QQ^V$ be a polytope and $y \in \QQ^V$.
  The \emph{separation problem} 
%\SEP{P,y} 
on $P$ 
is the problem of determining either
  (i) that $y \in P$ or (ii) a vector $c \in \QQ^V$ with $c^\top y
  > \max\condset{c^\top x}{x \in P}$ and $\|c\|_\infty = 1$.
\end{problem}

Over families of rational polytopes, the optimisation and separation problems
are polynomial-time equivalent (c.f., e.g., \cite[Theorem 6.4.9]{GLS88}).  Here
the time bound is measured in the size of the polytope and all other parameters of the problem.

\subsection{Representation}
When we deal with polytopes as objects in a computation, we need to
choose a representation which gives a finite
description of a polytope.  
In particular, in dealing with logical definability of
problems on polytopes, we need to choose a representation of polytopes
by relational structures.

\begin{definition}
A \emph{representation} of a class $\Pp$ of polytopes is a relational
vocabulary $\tau$ along with an \emph{onto} function
$\nu: \fin[\tau] \rightarrow \Pp$ which is isomorphism invariant,
that is, $\struct{A} \isom \struct{B}$ implies $\nu(\struct{A})
\isom \nu(\struct{B})$. 
\end{definition}

For concreteness, consider the vocabulary
$ \tau \defeq \vocmat \uplus \vocvec$ obtained by taking the
disjoint union of the vocabularies for rational matrices and vectors.
A $\tau$-structure over a universe consisting of a set $V$ of
variables and a set $C$ of constraints describes  a
constraint matrix $A \in \QQ^{C \times V}$ and bound vector
$b \in \QQ^C$.  Thus, the function taking such a structure to the
polytope $\poly A b$ is a representation of the class of rational
polytopes.  We call this the \emph{explicit representation}.  

Note that the explicit representation of polytopes has the property
that both the size of the polytope (i.e., the maximum size of any constraint) and the
number of constraints of $\nu(\struct A)$ are polynomially bounded in
the size of $\struct A$.  We will also be interested in
representations $\nu$ where the number of constraints in $\nu(\struct
A)$ is exponential in $|\struct A|$, but we always confine ourselves to
representations where the size of the constraints is bounded by a
polynomial in $\struct A$. 
We formalise this by saying that a representation
$\nu$ is \emph{well described} if there is a polynomial $p$ such that
$\SIZE{\nu(\struct A)} = p(|\struct A|)$, for all $\tau$-structures
$\struct A$. 
%% We will have occasion to consider polytopes that contain many constraints.
%% %%that are not well described.  
%% In particular, in Section~\ref{sec:matching} we consider
%% the maximum matching polytope which associates to a graph (i.e.\ a
%% structure in the vocabulary with one binary relation $E$) a polytope
%% where the number of constraints may be exponential in the size of the
%% graph.  However, since each individual constraint may be described using a number of bits polynomial in the size of the associated graph, the polytope is well described.
In particular, in all representations we consider 
the \emph{dimension} of the polytope $\nu(\struct A)$ is bounded by a
polynomial in $|\struct A|$. 

We are now ready to define what it means to express the linear
optimisation and separation
% and circumscription 
problems in \FPC.

\begin{definition}
\label{defn:expressing-optimisation-fpc}
We say that the linear optimisation problem for a class of polytopes $\Pp$ is
expressible in \FPC with respect to a representation
$\nu: \fin[\tau] \rightarrow \Pp$ if there is an \FPC interpretation of
$\vocnum \uplus \vocvec$ in $\tau \uplus \vocvec$ which takes a $\tau$-structure
$\struct A$ and a vector $c$ to a rational $f$ and vector $y$ such that either
(i) $f = 1$, and $\nu(\struct A)$ is unbounded in the direction of $c$, or (ii)
$f = 0$, and $\nu(\struct A) \neq \es$ iff $y \in \nu(\struct A)$ and $c^\top y
= \max\condset{c^\top x}{x \in \nu(\struct A)}$.
%%  which takes a structure coding a $\tau$-structure $\struct A$ and a vector
%%  $c$ to a number $f$ and a vector $y$ such that either (i) $f = 0$,
%%  $y \in \nu(\struct A)$ and $c^\top y \ge \max\condset{c^\top
%%  x}{x \in \nu(\struct A)}$, (ii) $f = 1$, $y = 0$ and $\nu(\struct A) = \es$,
%%  or (iii) $f = 1$, $y = 1$ and $\nu(\struct A)$ is unbounded in the direction
%%  of $c$. % (i.e., $\max\condset{c^\top x}{x \in P}$ is unbounded).
\end{definition}

\begin{definition}
\label{defn:expressing-separation-fpc}
The separation problem for a class of polytopes
$\Pp$ is expressible in \FPC with respect to a representation
$\nu: \fin[\tau] \rightarrow \Pp$ if there is an \FPC interpretation
of $\vocvec$ in $\tau \uplus \vocvec$ which takes a structure
coding a $\tau$-structure $\struct A$ and a vector $y$ to a vector $c$
such that either (i) $y \in \nu(\struct A)$ and $c=0$, or (ii) $c \in \QQ^V$ with $c^\top y
  > \max\condset{c^\top x}{x \in \nu(\struct A)}$ and $\|c\|_\infty = 1$.
\end{definition}

%% \begin{definition}
%% The circumscription problem for a class of polytopes
%% $\Pp$ is expressible in \FPC with respect to a representation
%% $\nu: \fin[\tau] \rightarrow \Pp$ if there is an \FPC interpretation
%% of $\vocnum$ in $\tau$ which takes a $\tau$-structure
%% to a rational $R$ such that $\nu(\struct A) \sse \sphere{0}{R}$.
%% \end{definition}

%% file: sep.tex
Let $A \in \QQ^{C \times V}$ be a constraint matrix and $b \in \QQ^C$ a
constraint vector of the polytope $P_{A,b}$.  
%The separation problem for the
%explicitly-represented polytope $P_{A,b}$ is easily solved.  Consider the
%algorithm $\Delta$ presented in Figure~\ref{fig:sep}.  It immediately follows
%from the definitions that this algorithm solves the separation problem on
%$P_{A,b}$.  
Figure~\ref{fig:sep} presents a 
straightforward algorithm ($\Delta$) for solving the
separation problem for the
explicitly-represented polytope $P_{A,b}$.  It is not hard to see
that the algorithm $\Delta$ can be implemented in 
time polynomial in the size of the explicit natural representation of inputs $A, b$ and $x$.

\algxio{$\Delta$}{fig:sep}{A separation oracle for explicitly-represented
rational polytopes.}{A,b,x}{$A \in \QQ^{C \times V}$, $b \in \QQ^C$ and
$x \in \QQ^V$.\\}{$c \in \QQ^V$ solving the separation problem for the polytope
$P_{A,b}$ and $x$.}{  
 \If{$A x \le b$} \Return $0^V$.\MEndIf \label{ln:check}
 \State Select $k \in C$ such that $A_k x > b_k$. \label{ln:choice}
 \State \Return $\frac{A_k}{\inorm{A_k}}$. \label{ln:cons_set}
}

If we try to express the algorithm $\Delta$ in fixed-point logic with counting,  we first note that
we can define in \FPC all the relevant manipulations on rational values, vectors
and matrices, such as norms, addition and multiplication~\cite{H10}, even when
they are indexed by unordered sets.  This shows that both
lines~\ref{ln:check} and \ref{ln:cons_set} of the algorithm can be simulated in
\FPC.  However, line~\ref{ln:choice} poses a problem 
as the logic is in general not able to \emph{choose} a particular element from
an unordered set.  Our key observation here is that linearity implies that the
sum of all such violated constraints is itself a violated constraint for non-empty
polytopes and hence the choice made by $\Delta$ is superfluous. 
This can be formally stated as follows.

\begin{prop}
  \label{prop:sum-constraint} Let $A \in \QQ^{C\times V}$, $b \in \QQ^{C}$,
  $x \in \QQ^{V}$ and $C \spe S \neq \es$.  Suppose $\poly A b$ is non-empty and
  $(Ax)_s \not\le b_s$ for all $s \in S$.  Define $a_S \de \sum_{s \in S} A_s$.
  Then $a_S^\top x > \max\condset{a_S^\top y}{y \in P_{A,b}}$ and $a_S \neq 0^V$.
\end{prop}

\begin{proof} Define $b_S \de \sum_{s \in S}  b_s$.
That $a_S^\top x > b_S$ is immediate from linearity.  Since the polytope is
non-empty pick any point $y \in P_{A,b}$.  By definition, $Ay \le b$.  Linearity
implies that $a_S^\top y \le b_S$.  Thus $a_S^\top x > b_S \ge \max\condset{a_S^\top
y}{y \in P_{A,b}}$.  This also implies that $a_S \neq 0^V$.
\end{proof}

This observation
leads to a definition in \FPC of the separation problem
for $\poly{A}{b}$ with respect to $x$.  Specifically, let $S \sse C$ be the set of constraints
which violate the inequality $Ax \le b$.  This set can be defined by a \FPC
formula using rational arithmetic.  If $S$ is empty, expressing $c = 0^V$
correctly indicates that $x \in P_{A,b}$.  Otherwise $S$ is non-empty; let $a_S$
be the sum of the constraints which $x$ violates.  Since the set $S$ is
definable in \FPC so is the sum of constraints indexed by $S$.  If $a_S \neq
0^V$, Proposition~\ref{prop:sum-constraint} implies that expressing $c$ as the
division of $a_S$ by its (non-zero) infinity norm correctly indicates a
separating hyperplane for $P_{A,b}$ through $x$; moreover, both operations are
in \FPC.  Otherwise, $a_S = 0^V$ and Proposition~\ref{prop:sum-constraint}
indicates that $P_{A,b}$ is empty.  This means that any non-zero vector defines
a separating hyperplane for $P_{A,b}$.  Thus it suffices for the interpretation
to express the vector $c = 1^V$. 
Overall, the above discussion gives us a proof of the following theorem.

\begin{thm}
  \label{thm:explicit-sep} There is an \FPC interpretation of $\vocvec$ in
  $\vocmat \uplus \vocvec \uplus \vocvec$ expressing the separation problem for
  the class of polytopes explicitly given by constraints and represented
  naturally as $\vocmat \uplus \vocvec$-structures.
\end{thm}

%% file: opt-sep.tex
In this section we present our main technical result, which is an \FPC reduction
from optimisation to separation, which treats the classical polynomial-time
reduction of the corresponding problems as a subroutine.  This classical result can be stated as follows.

%\marginnote{M: Perhaps discuss ``oracle polynomial time'' and ``blackbox''?}

\begin{thm}[{c.f., e.g., \cite[Theorem 6.4.9]{GLS88}}\footnote{The reverse of
this theorem also holds: An oracle for the linear optimisation problem can be
used to solve the separation problem.}] 
\label{thm:opt-to-sep-class}
The linear optimisation problem can be solved in polynomial time for any
well-described polytope given by a polynomial-time oracle solving the separation problem for
that polytope.
\end{thm}

\noindent
Below, we prove the following analogous result for fixed-point logic with counting.

\begin{thm}[Optimisation to Separation]
\label{thm:opt-to-sep}
  Let $\Pp$ be a class of well-described rational polytopes represented by
  $\tau$-structures and the function $\nu$.  Let $\Sigma$ be an \FPC
  interpretation of $\vocvec$ in $\tau \uplus \vocvec$ expressing the separation
  problem for $\Pp$ with respect to $\nu$.  Then there is an \FPC interpretation
  of $\vocnum \uplus \vocvec$ in $\tau \uplus \vocvec$ which expresses
  the linear optimisation problem for $\Pp$ with respect to $\nu$.
%% which takes a $\tau$-structure $\struct A$ and a vector $c$ to a rational $f$
%% and vector $y$ such that either (i) $f = 1$, and $\nu(\struct A)$ is
%% unbounded in the direction of $c$, or (ii) $f = 0$, $\nu(\struct A) \neq \es$
%% iff $y \in \nu(\struct A)$ and $c^\top y = \max\condset{c^\top
%% x}{x \in \nu(\struct A)}$.
\end{thm}

Observe that these theorems do not imply that every linear optimisation problem
can be solved in \FPC (or even in polynomial time).  Rather one can solve
particular classes of linear optimisation problems where domain knowledge can be
used to solve the separation problem.  We have the following generic consequence
in the case of explicitly-given polytopes when Theorem~\ref{thm:opt-to-sep} is
combined with Theorem~\ref{thm:explicit-sep}.

\begin{thm}[Explicit Optimisation]
\label{thm:explicit-opt}
There is an $\FPC$-interpretation of $\vocnum \uplus \vocvec$ in
$\vocmat \uplus \vocvec \uplus \vocvec$ expressing the linear optimisation
problem for the class of polytopes explicitly given by constraints and
represented naturally as $\vocmat \uplus \vocvec$-structures.
\end{thm}

The main idea behind the proof of Theorem~\ref{thm:opt-to-sep} is as follows.  
Suppose we are given a polytope $P \sse \QQ^V$ by 
an \FPC-interpretation $\Sigma_P$ that expresses the separation problem
for $P$.  A priori the elements of $V$ are indistinguishable.
However, $\Sigma_P$ may expose an underlying order in $V$ as it expresses
answers to the separation problem for $P$.  For example, suppose
$\Sigma_P$ on some input expresses a vector $d \in \QQ^V$ where the
components $d_u$ and $d_v$ for $u,v \in V$ are different.  This
information can be used to distinguish the components $u$ and $v$; moreover, it
can be used to order the components because $d_u$ and $d_v$ are distinct elements of a field
with a total order.  As $\Sigma_P$ is repeatedly used it may expose
more and more information about the asymmetry of $P$.  This partial
information can be represented by maintaining a sequence of
equivalence classes $(V_i)_{i=1}^k$ partitioning $V$.  This
equivalence relation is progressively refined through further
invocations of the separation oracle.  Initially all elements of $V$
reside in a single class.

It is natural to consider the polytope $P'$ derived from $P$ by taking
its quotient under the equivalence relation defined in this way.  
Intuitively, this maps polytopes in $\QQ^V$ to
polytopes in $\QQ^k$ by summing the components in each equivalence class to form
a single new component which is ordered by the sequence.  We call this
process \emph{folding}.  We observe that a separation oracle for $P'$ can be
constructed using $\Sigma_P$, provided the answers of $\Sigma_P$ never expose
more asymmetry than was used to derive $P'$.  However, failing to meet this
proviso is informative---it further distinguishes the elements of $V$---and
refines the sequence of equivalences classes.

These observations suggest the following algorithm.  Start with a sequence
$(V_1)$ of exactly one class which contains all of $V$.  Construct the
folded polytope $P'$ with respect to this sequence and the associated separation
oracle $\Sigma_{P'}$ from $\Sigma_P$.  Attempt to solve the linear optimisation
problem on the folded polytope (which lies in an ordered space) using the
Immerman-Vardi theorem \cite{V82,I86} and the classical polynomial-time
reduction from optimisation to separation (Theorem~\ref{thm:opt-to-sep-class})\footnote{Recall that, by the Immerman-Vardi theorem, every polynomial-time property of ordered structures is definable in fixed-point logic, and hence also in \FPC.}.
Should $\Sigma_P$ at any point answer with a vector that distinguishes more
elements of $V$ than the current sequence of equivalence classes, then we: 
(i) abort the run; 
(ii) refine the equivalence classes and the folded polytope with this new information; and, finally, 
(iii) restart the optimisation procedure on this more representative problem instance.
Since the number of equivalence classes
increases each time the algorithm aborts, it eventually solves
the optimisation problem for some $P'$ without aborting. We argue that this solution
for $P'$ can be translated into a solution for $P$.

A key aspect of this approach is that it treats the polynomial-time reduction
from optimisation to separation as a blackbox, i.e., it assumes nothing about
how the reduction works internally.\footnote{It is, in fact, possible to
translate the classical reduction from optimisation to separation line by line
into \FPC (in the spirit of Section~\ref{sec:sep}).  However, this translation
quickly becomes mired in intricate error analysis which is both tedious and
opaque.}  Before formally describing the algorithm we establish a number of
useful definitions and technical properties.

%%-----------------------------------------------------------------------------

\subsection{Folding}

Let $V$ be a set.  For $k \le |V|$, let $\sigma : V \rightarrow [k]$
be an onto map.  We call $\sigma$ an \emph{index map}.  For $i \in [k]$
define $V_i \de \condset{s \in V}{\sigma(s) = i}$.  The sequence of
sets $V_i$ is a partition of $V$.

\begin{definition}[Folding]\label{def:folding} ~\\
For a vector $x \in \QQ^V$ let the \emph{almost-folded} vector $\afd{x}$ of
$\QQ^{k}$ be given by $$(\afd{x})_i \de \sum_{v \in V_i} x_v, \text{ for } i \in [k].$$ For a vector $x
\in \QQ^V$ let the \emph{folded} vector $\fd{x}$ of $\QQ^{k}$ be given
by $$(\fd{x})_i \de \frac{\afd{x}}{|V_i|}, \text{ for } i \in [k].$$ For a vector $\order{x} \in \QQ^{k}$
let the \emph{unfolded} vector $\fold{\order{x}}{-\sigma}$ of $\QQ^V$ be given
by $$(\fold{\order{x}}{-\sigma})_v \de \order{x}_i, \text{ with } V_i \ni v, \text{ for } v \in V.$$
\end{definition}

We say a vector $x \in \QQ^V$ \emph{agrees with} $\sigma$ when for all $v,v' \in
V$, $\sigma(v) = \sigma(v')$ implies $x_v = x_{v'}$.  It easily follows that if
$x$ agrees with $\sigma$ then $\fold{\fold{x}{\sigma}}{-\sigma} = x$.  When
vectors agree with $\sigma$ and $\sigma$ is clear from context we often use the
font, as above with $x$ and $\order{x}$, to indicate whether a vector is
unfolded and lies in $\QQ^V$, or folded and lies in $\QQ^k$, respectively.  The
notion of folding naturally extends to a set $S \sse \QQ^V$ (and hence
polytopes): Let $\fd{S} \de \condset{\fd{s}}{s \in S}$.  See
Figures~\ref{fig:proj1}~and~\ref{fig:proj2} for examples of folding
polytopes.  Note that $\fd{P}$ is a projection of $P$ into the
$k$-dimensional space $\QQ^k$.

\begin{figure}
\centering
\includegraphics{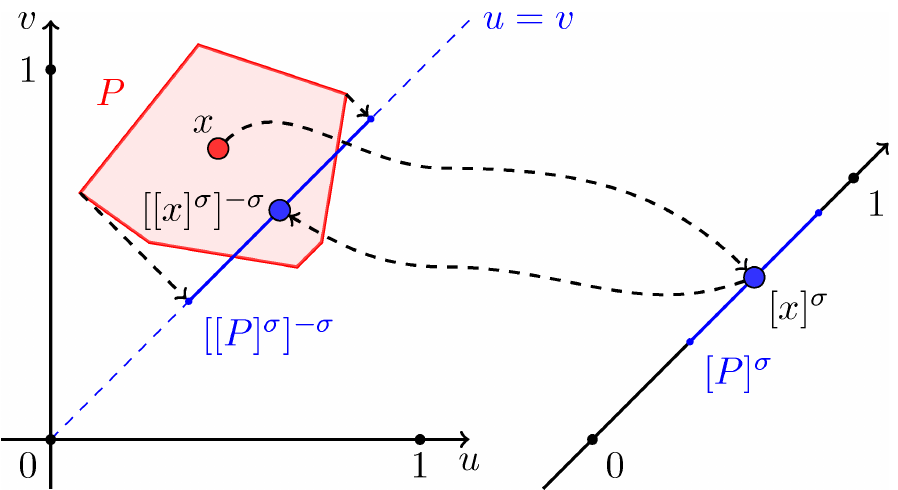}
\caption{Folding and unfolding a polytope $P \sse \QQ^{\set{u,v}}$ with respect to $\sigma
= \set{u \rightarrow 0, v \rightarrow 0}$.\label{fig:proj1}}

\includegraphics{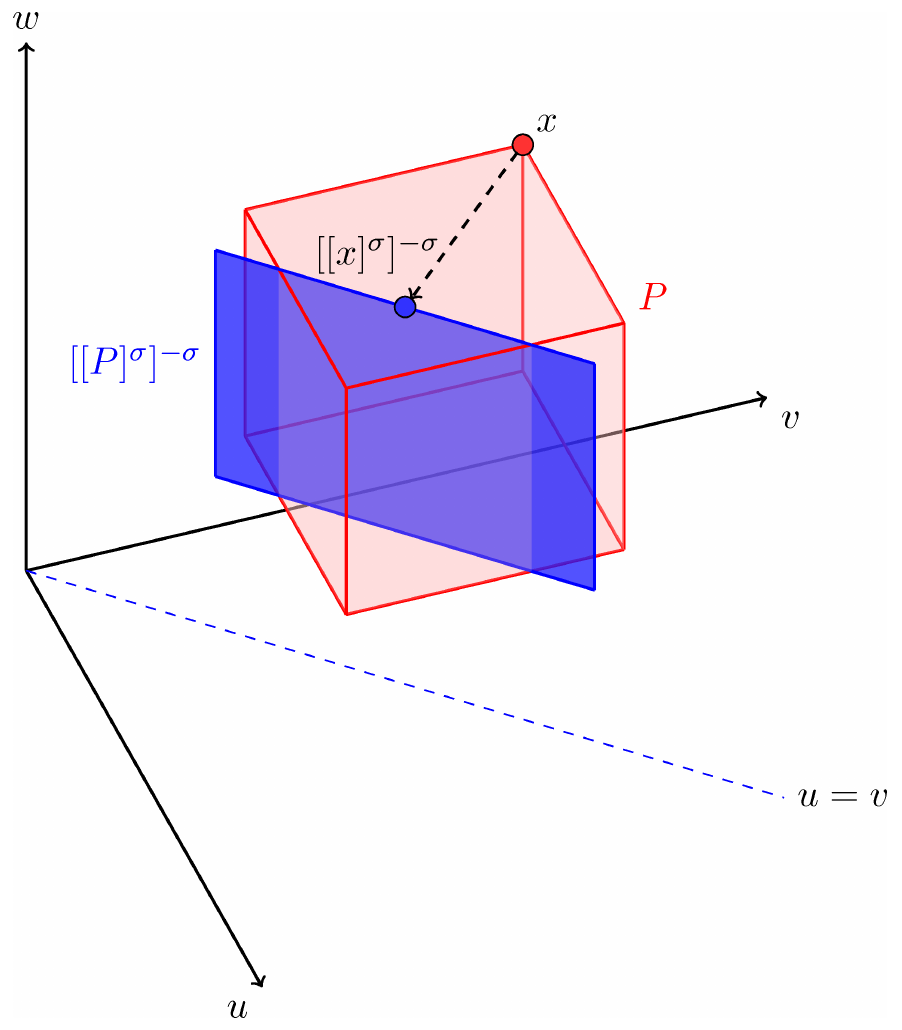}

\caption{Folding and unfolding a polytope $P \sse \QQ^{\set{u,v,w}}$ with respect
to $\sigma = \set{u \rightarrow 0, v \rightarrow 0, w \rightarrow 1}$.}
\label{fig:proj2}

\end{figure}

Several useful properties of folding and unfolding follow directly from their
definitions. 

\begin{prop}
  \label{prop:fold} Let $\sigma: V \rightarrow [k]$ be an index map and
  $c,x \in \QQ^V$ such that $c$ agrees with $\sigma$. Then, $$c^\top \unfd{\fd{x}} = c^\top
  x = \afd{c}^\top\fd{x}.$$
\end{prop}

\begin{proof} We begin by proving the first equality.  Fix $i \in [k]$.  Definition~\ref{def:folding} implies that
  \begin{equation} \label{eqn:fold}
  \begin{aligned}
  \sum_{v \in V_i} (\unfd{\fd{x}})_v - x_v &= \sum_{v \in V_i} (\fd{x})_i
  - \sum_{v \in V_i} x_v = \sum_{v \in V_i} \left(\frac{1}{|V_i|}\sum_{v' \in V_i} x_{v'}\right) - \sum_{v \in V_i}
  x_v \\
  &= \frac{|V_i|}{|V_i|} \sum_{v' \in V_i} x_{v'} - \sum_{v \in V_i} x_v  = 0. \\
  \end{aligned}
  \end{equation}

We conclude that  
  \alignedeq{
  c^\top (\unfd{\fd{x}} - x) &= \sum_{v \in V} c_v ((\unfd{\fd{x}})_v - x_v) \\
  &= \sum_{i \in [k]} \sum_{v \in V_i} c_v ((\unfd{\fd{x}})_v - x_v) & V =
  \uplus_{i \in[k]} V_i \\
  &= \sum_{i \in [k]} \sum_{v \in V_i} \frac{(\afd{c})_i}{|V_i|} ((\unfd{\fd{x}})_v
  - x_v) & c \text{ agrees with } \sigma \\
  &= \sum_{i \in [k]} \frac{(\afd{c})_i}{|V_i|} \sum_{v \in V_i} \unfd{(\fd{x}})_v -
  x_v & \text{linearity} \\
  &= \sum_{i \in [k]} \frac{(\afd{c})_i}{|V_i|} \cdot 0 = 0. & \text{\eqref{eqn:fold}} \\
  }

We now argue the second equality.

  \alignedeq{
  c^\top x &= \sum_{v \in V} c_v x_v \\
  &= \sum_{i \in [k]} \sum_{v \in V_i} c_v x_v & V = \uplus_{i \in [k]} V_i\\
  &= \sum_{i \in [k]} \sum_{v \in V_i} \left(\frac{1}{|V_i|} \sum_{v' \in V_i}
  c_{v'}\right) x_v & c \text{ agrees with } \sigma \\
  &= \sum_{i \in [k]} \left(\sum_{v' \in V_i} c_{v'}\right)
  \left(\frac{1}{|V_i|} \sum_{v \in V_i} x_v \right) & \text{linearity} \\
  &= \sum_{i \in [k]} (\afd{c})_i (\fd{x})_i & \text{Def.~\ref{def:folding}} \\
  &= \afd{c}^\top \fd{x}. \\
  }
\end{proof}

\subsection{Folding Polytopes}

The diagrams in Figures~\ref{fig:proj1}~and~\ref{fig:proj2} suggest intuitively
that the result of folding a polytope is itself a polytope; the following
proposition makes this connection concrete.

\begin{prop}
  \label{prop:poly} Let $P$ be a polytope in $\QQ^V$ and let $\sigma:
  V \rightarrow [k]$ be an index map.  Then the folded set $\fd{P}$ is a
  polytope with $\fSIZE{\fd{P}} \le 48k^3 |V|^3 \fSIZE{P}$.
\end{prop}

\begin{proof}
Let $P = \conv(S_1) + \cone(S_2)$ for two finite sets of points
$S_1,S_2 \sse \QQ^V$.  By the linearity of $\fd{\cdot}$ we have 
\begin{align*}
  \fd{P}
    &= \bigfd{\conv(S_1)+\cone(S_2)} \\
    &= \fd{\conv(S_1)} + \fd{\conv(S_2)} \\
    &= \conv(\fd{S_1}) + \cone(\fd{S_2}).
\end{align*}
\noindent
We conclude that $\fd{P}$ is a
polytope.  We have $$\fSIZE{\fd{P}} \le 3k^2\vSIZE{\fd{P}} \le 3k^2 \cdot
4k|V|\vSIZE{P} \le 12k^3|V| \cdot 4|V|^2 \fSIZE{P}$$ where the middle
inequality comes from bounding the bit complexity of $\frac{1^\top v}{|V|}$
for extremal vertices $v \in P$.
\end{proof}

For a polytope $P \sse \QQ^V$ and a point $x \in \QQ^V$ (with
$x \not\in P$) we say that \emph{all
separating hyperplanes at $x$ disagree with} $\sigma$ if
there is no $c \in \QQ^V$ which both agrees with $\sigma$ and has $c^\top x >
\max\condset{c^\top y}{y \in P}$.  This
induces an alternative characterisation of the polytope $\fd{P}$.

\newcommand{\pctext}[2]{\text{\parbox{#1}{\centering #2}}}

\begin{lem} 
\label{prop:altpoly}
Let $P$ be a polytope in $\QQ^V$ and $\sigma: V \rightarrow [k]$ be an
index map. Then 
\begin{alignat*}{2}
  \fd{P} = P' \de     
    \biggl\{ 
      \order{x} \in \QQ^k 
    &&\;\bigg|\; 
      \pctext{2.5in}{$\unfd{\order{x}} \in P$ or all separating hyperplanes at $\unfd{\order{x}}$ disagree with $\sigma$}
    \biggr\}.
\end{alignat*}

% \begin{align*}
%   \fd{P} = P' \de \{
%     &\order{x} \in
%     \QQ^k \;|\;
%     \unfd{\order{x}} \in P \\
%     &\text{ or all separating hyperplanes at
%   } \unfd{\order{x}} \text{ disagree with } \sigma
%   \}.
% \end{align*}
\end{lem}

\begin{proof}
  We show both inclusions.

  \case{1.}{$\fd{P} \sse P'$:} 
  Let $\order{x} \in \fd{P}$.  By definition there is a point $x \in P$ such
  that $\fd{x} = \order{x}$.  Suppose $x = \unfd{\order{x}}$, then
  $\unfd{\order{x}} \in P$ and hence $\order{x} \in P'$.  Thus assume
  $\unfd{\order{x}} \neq x$.  Let $c \in \QQ^V$ be any vector agreeing with
  $\sigma$.  By Proposition~\ref{prop:fold} we have $c^\top \unfd{\order{x}} =
  c^\top \unfd{\fd{x}} = c^\top x.$ Since $x \in P$, $c$ is not the normal of a
  separating hyperplane through $\unfd{\order{x}}$.  We conclude that all
  separating hyperplanes through $\unfd{\order{x}}$ disagree with $\sigma$ and
  hence that $\order{x} \in P'$.
  
  \medskip

  \case{2.}{$\fd{P} \spe P'$:} 
  Let $\order{x} \in P'$.  Suppose $\unfd{\order{x}} \in P$, then $\order{x}
  = \fd{\unfd{\order{x}}} \in \fd{P}$.  Thus assume that
  $\unfd{\order{x}} \not\in P$ and that all separating hyperplanes through
  $\unfd{\order{x}}$ disagree with $\sigma$.  This means that for any vector
  $c \in \QQ^V$ that agrees with $\sigma$ the hyperplane through
  $\unfd{\order{x}}$ with normal $c$ intersects $P$ and thus there is a point
  $y \in P$ which has $c^\top \unfd{\order{x}} = c^\top y$.  This further
  implies that $c^\top \unfd{\order{x}} \le \max \condset{c^\top y}{y \in P}$. 
  Since $c$ agrees with $\sigma$, Proposition~\ref{prop:fold} implies that
  $$\afd{c}^\top \order{x} \le \max\condset{\afd{c}^\top \fd{y}}{y \in P}.$$
  Observe $\afd{\condset{c \in \QQ^V}{c \text { agrees with } \sigma}} = \QQ^k$.
  This means for any vector $c' \in \QQ^k$,
  $c'^\top \order{x} \le \max\condset{c'^\top \fd{y}}{y \in P}$.  In particular,
  for every constraint defining the polytope $\fd{P}$, 
  $\order{x}$ also satisfies that constraint.  We conclude that
  $\order{x} \in \fd{P}$.
\end{proof}

\subsection{Expressing Optimisation in \FPC}

Suppose we are given a polytope $P \sse \QQ^V$ via a separation oracle
$\Delta_P$, and a vector $c$ indicating a linear objective.  The algorithm
maintains an index map $\sigma : V \rightarrow [k]$ that indicates a sequence of
equivalence classes of $V$ which have not been distinguished by the algorithm so
far.  Initially this index map is given by ordering variables according to their
relative values in $c$.  Under the assumption that $\sigma$ accurately describes
the symmetries of $P$ we execute the polynomial-time reduction from optimisation
to separation on the polytope $\fd{P}$ and objective $\afd{c}$.  Since $\fd{P}$
lies in an ordered space, it follows from the Immerman-Vardi theorem that the reduction can be
expressed in fixed-point logic with counting.

To this end, a separation oracle $\Delta_{\fd{P}}$ must be specified for
the polytope $\fd{P}$.  Given a point $\order{x} \in \QQ^k$, we argue that the
result of applying $\Delta_P$ to the unfolding of $\order{x}$ either determines
the point is in $P$, and hence also in $\fd{P}$; or determines a separating
hyperplane for $P$.  If a separating hyperplane is determined, it can be folded
into a separating hyperplane for $\fd{P}$, but only if the hyperplane normal
agrees with $\sigma$.  In the case the separating hyperplane disagrees with
$\sigma$, our assumption about $P$ is violated, and our separation oracle does
not have enough information to proceed.  Indeed, folding the resulting normal
may produce $0^k$ which is not a valid answer.  In this case, the algorithm
aborts the run of the linear optimisation algorithm, and returns the disagreeing
hyperplane normal.  The algorithm then combines the disagreeing normal with its
current index map $\sigma$ to produce a new index map which is consistent with
$\sigma$ and agreed with by the disagreeable hyperplane normal.  This strictly
increases the number of equivalence classes of variables induced by the index
map.  The above procedure can abort at most $|V|$ times before $\sigma$ exactly
characterises the order of $V$ relative to $P$.  After this point the linear
optimisation algorithm cannot abort and hence must solve the optimisation
problem for $\fd{P}$ which can be unfolded into a solution for $P$.

With this intuition in mind the formal proof is as follows.

\begin{proof}[Proof of Theorem~\ref{thm:opt-to-sep}]

For completeness the entire algorithm \FPCOpt is described in
Figure~\ref{fig:fpcopt}. The algorithm uses two subroutines \Refine and \Opt.
The subroutine \Refine$(\sigma,d)$ takes as input an index map $\sigma$ of $V$ represented
in $\NN^V$ and a vector $d \in \QQ^V$ and computes a new index map $\sigma'$
with the following two properties:
\begin{itemize}
\item for all $v,v' \in V$ with $\sigma(v) < \sigma(v')$, $\sigma'(v) <
  \sigma'(v')$, and
\item for all $v,v' \in V$ with $\sigma(v) = \sigma(v')$, $\sigma'(v) <
  \sigma'(v')$ iff $d_v < d_{v'}$.
\end{itemize}
 
\algxio{\FPCOpt}{fig:fpcopt}{An instrumentation of the reduction from
  optimisation to separation.}{P,\Delta_P,c}{ 
  \begin{itemize}
  \item A well-described polytope $P \sse \QQ^V$ with a separation
      oracle $\Delta_{P}$, and
  \item a linear objective $c \in \QQ^V$.
  \end{itemize}
}{
  \begin{itemize}
    \item $f = 1$ and $y = 0^V$, if $P$ is unbounded along $c$; or otherwise
    \item $f = 0$ and $y \in \QQ^V$, s.t. if $P \neq \es$ then $y \in P$ and 
      $c^\top y = \max\condset{c^\top x}{x \in P}$.
  \end{itemize}\vspace{-2.75ex}
}{
 \State $\sigma \gets \Refine(0^V,c)$.
 \While{true}

   \State $(f,\order{x}) \gets \Opt(\fd{P},\Delta_{\fd{P}},\afd{c})$.
   \If{aborted with $\sigma'$} %\com{FIXME - Better exception syntax? -- I think it's fine (BH)}
   \State $\sigma \gets \sigma'$. 
   \Else
   \State $\sigma' \gets \Refine(\sigma,\Delta_{P}(\unfd{\order{x}}))$.
   \If{$\sigma \neq \sigma'$}
   \State $\sigma \gets \sigma'$.
   \Else
   \State \Return $(f,\unfd{\order{x}})$.
   \EndIf
   \EndIf
 \EndWhile
 %% \medskip
 %% \hrule
 \medskip
 \hrule
 \medskip

 \Procedure{$\Delta_{\fd{P}}$}{$\order{x}$} \label{alg:fpcopt:oracle-begin}
     \State $d \gets \Delta_{P}(\unfd{\order{x}})$.
     \State $\sigma' \gets \Refine(\sigma,d)$.
     \If {$\sigma \neq \sigma'$} \textbf{abort} with $\sigma'$. \MEndIf
     \State \Return $\afd{d}$.
   \EndProcedure \smallskip \label{alg:fpcopt:oracle-end}
}

It is straightforward to observe that when $\Refine(\sigma,d)$ produces an
index map $\sigma'$ which is different from $\sigma$, then $\sigma'$ induces
strictly more equivalence classes on $V$ then $\sigma$ does.  Clearly, no
index map can induce more than $|V|$ equivalences classes.  The
subroutine \Opt{} solves the linear optimisation problem on an ordered space $\QQ^k$
with a given linear objective and a polytope given by a separation oracle.
Without loss of generality assume \Opt{} returns a integer-vector pair $(f,y)$
which is $(1,0^k)$ when the objective value is unbounded and $(0,y)$ when
$y \in \QQ^k$ is an optimal point in the polytope if, and only if, the polytope is non-empty.

We first argue that the algorithm is correct, assuming the correctness of \OPT
and \Refine.  For any index map $\sigma: V \rightarrow [k]$, $\fd{P}$ is a
polytope by Proposition~\ref{prop:poly}.  We show that the procedure
$\Delta_{\fd{P}}$ described in lines~\ref{alg:fpcopt:oracle-begin}
to \ref{alg:fpcopt:oracle-end} acts as a separation oracle for $\fd{P}$ provided
the answer given by the separation oracle $\Delta_P$ agrees with
$\sigma$.  If $\Delta_{P}(\unfd{\order{x}})$ outputs $d = 0^V$, then this indicates
that $\unfd{\order{x}} \in P$, and hence $\order{x} \in \fd{P}$ by
Proposition~\ref{prop:altpoly}.  Trivially $0^V$ agrees with $\sigma$, so
$\afd{d} = \afd{0^V} = 0^k$ is returned by $\Delta_{\fd{P}}$ correctly indicating
that $\order{x} \in \fd{P}$.  Otherwise, $d \neq 0^V$ and indicates that
$\unfd{\order{x}} \not\in P$ but $d^\top \unfd{\order{x}} > \max\condset{d^\top
y}{y \in P}$.  If $d$ agrees with $\sigma$ we have, by
Proposition~\ref{prop:fold}, $\afd{d}^\top \order{x}
> \max\condset{\afd{d}^\top \fd{y}}{y \in P}$.  This is equivalent to
$\afd{d}^\top \order{x}
> \max\condset{\afd{d}^\top \order{y}}{\order{y}\in \fd{P}}$.  Hence
$\afd{d}^\top$ is the normal of a separating hyperplane of $\fd{P}$ through
$\order{x}$.  Since $d$ agrees with $\sigma$, $\sigma' = \sigma$ and $\afd{d}$
is correctly returned.  If $d$ does not agree with $\sigma$, then \Refine
produces a $\sigma' \neq \sigma$ and the procedure aborts.  We conclude that \textit{(i)}
when $\Delta_{\fd{P}}$ does not abort it behaves as a separation oracle for
$\fd{P}$, and \textit{(ii)} when $\Delta_{\fd{P}}$ aborts the returned index map
$\sigma'$ is a strict refinement of $\sigma$.  Thus $\Delta_{\fd{P}}$ is a
separation oracle for $\fd{P}$, provided it does not abort.  When \OPT runs on
$\Delta_{\fd{P}}$ without aborting the result must be a solution to the linear
optimisation problem on $\fd{P}$.

Let $\order{x} \in \fd{P}$ be such that $\afd{c}^\top \order{x} \ge
\max\condset{\afd{c}^\top \order{y}}{\order{y} \in \fd{P}}$, i.e., it is a
solution to the linear optimisation problem on $\fd{P}$ along $\afd{c}$.  By
Proposition~\ref{prop:altpoly} this means that either \textit{(i)} $\unfd{\order{x}} \in
P$ or \textit{(ii)} $\Delta_{P}(\unfd{\order{x}})$ must disagree with $\sigma$.  Applying
$\Delta_P$ to $\unfd{\order{x}}$ distinguishes these two cases.  In case \textit{(i)},
$\afd{c}^\top \order{x} = c^\top \unfd{\order{x}} \ge \max\condset{c^\top y}{y
  \in P}$ by Proposition~\ref{prop:fold}, because the initialisation of $\sigma$
forces $c$ to agree with $\sigma$.  This means that $\unfd{\order{x}}$ is a
solution to the linear optimisation problem for the polytope $P$ and the
objective $c$.  In case \textit{(ii)}, $\unfd{\order{x}} \not\in P$ but
$\Delta_P(\unfd{\order{x}})$ is guaranteed to improve the index map.  In the case
that the linear optimisation algorithm returns that $\fd{P}$ is unbounded in the
direction of $\afd{c}$, it implies, via similar analysis, that $P$ is unbounded
in the direction $c$.  Finally, when the optimisation algorithm reports that
$\fd{P}$ is empty we conclude that $P$ must be empty as well, because if $P$
contains at least one point then $\fd{P}$ must also contain at least one point.  The
algorithm correctly translates the solutions for the linear optimisation problem for
$\fd{P}$ back to solutions for $P$.  This means that when \FPCOpt returns its
result is correct.

We now observe that this algorithm runs in polynomial time.  The main
loop cannot execute more than $|V|$ times, because, as established above, at
each step either the index map $\sigma$ is improved to induce more equivalence
classes---up to $|V|$ classes---or the algorithm returns a correct solution to
the linear optimisation problem on $P$.  The size of all of the objects referred
to by the algorithm can be polynomially bounded by a function of the input
length.  In particular, since $P$ is well-described by $\Delta_P$, there is a
polynomial bound on its bit complexity and this induces a bound on the size of
$\fd{P}$ through Proposition~\ref{prop:poly} and implies that $\fd{P}$ is well
described.  This implies that the bit complexity of values in the algorithm can
be bounded by some fixed polynomial.  This means that folding and unfolding can
be computed in polynomial time.  Similarly, a naive implementation of the
subroutine \Refine can be seen to run in polynomial time in $|V|$ and the bit
complexity of its input rational vector.  Since $\fd{P}$ is a well-described
polytope with a polynomial-time separation oracle $\Delta_{\fd{P}}$ we can use
the polynomial-time algorithm for \OPT from
Theorem~\ref{thm:opt-to-sep-class} to solve the linear optimisation problem on
$\fd{P}$.  Combining all these parts implies that \FPCOpt is a
polynomial-time algorithm.

We conclude by arguing that the behavior of \FPCOpt can be simulated in \FPC.
Relative to an index map $\sigma$ expressible in \FPC, folding and unfolding can
be expressed in \FPC using basic rational arithmetic.  It is similarly routine
to express \Refine in \FPC by defining the equivalence classes and then counting
sizes to determine the correct position of each equivalence class relative to an
\FPC-definable $\sigma$ and vector.  Moreover, there is a \FPC-interpretation
$\Sigma_{\fd{P}}$ expressing the separation problem for $\fd{P}$.  This implies
there is an \FPC-interpretation for the combination of \Opt{} and the separation oracle given by
$\Sigma_{\fd{P}}$, because the polytope $\fd{P}$ lies in an ordered space and
the Immerman-Vardi theorem \cite{V82,I86} indicates that any polynomial-time
property of ordered structures can be defined in fixed-point logic (and hence
in \FPC).  It is easy to see that the algorithm's main loop and control
structure can be simulated in \FPC.  Combining everything gives
an \FPC-interpretation simulating \FPCOpt and hence expressing the linear
optimisation problem for $P$ given a \FPC-interpretation expressing the
separation problem for $P$.
\end{proof}

\noindent
In the next four sections we demonstrate a number of applications of Theorem~\ref{thm:opt-to-sep} for
expressing classical combinatorial optimisation problems in fixed-point logic with counting.

%% file: flows.tex
Let $G = (V,c)$ be a graph with non-negative edge capacities, that is, $c : V
\times V \rightarrow \nnQQ$. % and for all $u,v \in V$, $c(u,v) = c(v,u)$.  
For a pair of distinct vertices $s,t \in V$ an \emph{\fl}is a
function $f : V\times V \rightarrow \nnQQ$ satisfying capacity constraints
$0 \le f(u,v) \le c(u,v)$ on each pair of distinct $u,v \in V$ and conservation
constraints $\sum_{v \in V}(f(v,u) - f(u,v)) = 0$ on all vertices $u \in
V\bs\set{s,t}$.  The \emph{value} $\val(f)$ of the flow $f$ is simply the
difference in in-flow and out-flow at $t$, i.e., $\sum_{v \in V}(f(v,t) -
f(t,v))$.  Observe that any flow $f$ can be \emph{normalised} to $f'$ so that for any pair
of distinct $u,v \in V$ at least one of $f'(u,v)$ and $f'(v,u)$ is zero (i.e., if
$f(u,v) \ge f(v,u)$, set $f'(u,v) \de f(u,v)-f(v,u)$ and $f'(v,u) \de 0$;
obviously this preserves the capacity constraints, the conservation constraints
and the value of the flow).  
%Without loss of generality we assume that all flows
%are normalised.  
A \emph{maximum \fl of $G$} is a flow whose value is maximum
over all \fls.   

%The \emph{maximum flow of $G$} is the maximum \fl over all
%choices of distinct vertices $s,t$.

\begin{obs}
  \label{obs:2} Fix $G = (V,c)$ and $s,t \in V$.  Let $f_1, f_2$ be two \fls in
  $G$.  Fix any $\alpha \in \QQ$ with $0 \le \alpha \le 1$, let $f'
  := \alpha \cdot f_1 + (1-\alpha)\cdot f_2$.  Then $f'$ is an \fl of $G$ and
  $\val(f') = \alpha \cdot \val(f_1) + (1-\alpha)\cdot \val(f_2)$.  In
  particular, if $f_1$ and $f_2$ are maximum \fls then so is any convex
  combination $f'$.
\end{obs}

Let $G|_f \de (V,c - f)$ denote the \emph{residual graph} of $G$ with respect to
the flow $f$.  %Observe that in $G|_f$ capacities are no longer symmetric.
The standard formulation of the maximum \fl problem as a linear program is as follows:
\begin{equation}
  \label{eqn:flowLP}
  \begin{aligned}
    \max& \sum_{v \in V} (f(v,t) - f(t,v)) \quad\quad \text{ subject to} \\[1.25ex] &
    \sum_{v \in V} (f(v,u) - f(u,v)) = 0,\;\; \forall u \in V \bs \set{s,t}
    \\ & 0 \le f(u,v) \le c(u,v),\;\; \forall u \neq v \in V.
  \end{aligned}
\end{equation}

\subsection{Expressing Maximum Flow in \FPC}

Observe that there are $|V|(|V|-1)$ variables in linear program
\eqref{eqn:flowLP} corresponding to $f(u,v)$ for distinct $u,v \in V$.  The
program has $2|V|^2 - 4$ constraints.  Both the variables and constraints can be
indexed by tuples of elements from $V$.  It can easily be established that the
maximum \fl linear program can be defined by an \FPC interpretation.  That is to
say, suppose that a capacitated graph $(V,c)$ is given as a $\vocmat$-structure
with universe $V$ where the rational %symmetric 
matrix $c \in \nnQQ^{V \times V}$ codes the
capacities.  Then, there is an \FPC interpretation from $\vocmat$ to
$\vocmat \uplus \vocvec$ that takes a capacitated graph $(V,c)$ and a pair
$s,t \in V$ and explicitly expresses a constraint matrix $A$ and vector $b$
encoding the corresponding flow polytope.  The flow polytope is bounded because
each variable is constrained from both above and below.  Further the flow
polytope is nonempty because the capacities in $G$ are nonnegative and hence the
zero flow is a member of the polytope.  Thus, because this polytope is explicit,
Theorem~\ref{thm:explicit-opt} immediately gives an \FPC interpretation
expressing the optimisation problem on the flow polytope.

\begin{thm}
  \label{thm:flow} There is an \FPC interpretation $\Phi(s,t)$ of $\vocmat$ in
  $\vocmat$ which takes a $\vocmat$-structure coding a capacitated
  graph $G$ to a $\vocmat$-structure coding a maximum \fl of $G$.
\end{thm}

Note that as the interpretation $\Phi$ defines a particular flow, the flow must,
in some sense, be canonical because it is produced without making any choices.
Informally, it is a convex combination of maximum flows resulting from the
consideration of all orderings consistent with the most refined index map
determined by the \FPC interpretation of Theorem~\ref{thm:opt-to-sep}.  This is
possible because of Observation~\ref{obs:2}.  In our remaining
applications---minimum cut and maximum matching---the analog of
Observation~\ref{obs:2} does not hold: Convex combinations of cuts or matchings
are not necessarily cuts or matchings.  In the former it is still possible to
define the notion of a canonical optimum.  In the latter case it is easy to
observe, as noted in the introduction, that defining a canonical maximum
matching is not possible.

%% file: cuts.tex
%% We now show that Theorem~\ref{thm:flow} can be used to define minimum cuts in
%% undirected graphs. 

An \emph{\cu}of a capacitated graph $G = (V,c)$ is a subset $C$ of the vertices
$V$ which contains $s$ but not $t$.  The \emph{value} $\val(C)$ of the cut $C$
is the sum of the capacity of edges going from vertices in $C$ to vertices in
$V \bs C$.  A \emph{minimum \cu of $G$} is a cut whose value is the minimum
over all \cus.  A \emph{minimum cut of $G$} is a minimum \cu over all choices of
distinct vertices $s,t$.  By the max-flow/min-cut theorem, a maximum \fl and a
minimum \cu have the same value.  This duality allows the construction of
minimum cuts from maximum flows.  In this section we describe an \FPC formula
defining a minimum \cu in a graph using the \FPC interpretation for
the maximum \fl problem given by Theorem~\ref{thm:flow}; we show this minimum
cut is canonical in a strong sense.

\subsection{Expressing Canonical Minimum Cut in \FPC}

First, we define a notion of directed reachability in capacitated
graphs.  A vertex $v$ is \emph{reachable} from a vertex $u$ if there is a path
in the graph which follows directed edges with non-zero capacity (this is
exactly directed reachability in the graph induced by eliminating zero capacity
edges).  Let $f$ be a maximum \fl in $G=(V,c)$ with normalised flow $f'$.  Define
$C_f \de \condset{v \in V}{v \text{ reachable from } s \text{ in } G|_{f'}}$.
$C_f$ is a minimum \cu in $G$.  Since $f'$ is normalised, every edge leaving $C_f$ must be at
full capacity in $f'$.

Given the \FPC interpretation $\Phi$ from Theorem~\ref{thm:flow} expressing
an \fl $f$, it is not difficult to construct a formula of \FPC which defines the normalised flow $f'$ and then the
set of vertices $C_f$.

\begin{thm}
  \label{thm:cut} There is a formula $\xi(x,s,t)$ of $\FPC$ which given a
  $\vocmat$-structure coding a capacitated graph $G = (V,c)$, defines the
  vertices in a minimum \cu of $G$.
\end{thm}

In fact, the cut $C_f$ does not depend on $f$ at all, as we show next.  Indeed,
$C_f$ is the smallest minimum \cu in the sense that it is contained in all other
minimum \cus of $G$.

\begin{lem} 
  \label{lem:1} Let $G = (V,c)$ be a capacitated graph with distinct
  vertices $s,t \in V$.  Then the cut $C_f$ is independent of the choice of a
  maximum \fl $f$ of $G$.  Moreover, $C_f$ is the intersection of all
  minimum \cus of $G$.
\end{lem}

\begin{proof}[Proof of Lemma~\ref{lem:1}]
    Suppose not.  There are two distinct minimum \cus $C \de C_{f}$ and
    $C' \de C_{f'}$ with corresponding normalised \fls $f$ and $f'$.  Since $C$ and
    $C'$ are different there exists, without loss of generality, $v \in C' \bs
    C$.  Consider the flows through $\bar{C} \cap C'$.  We use
    $a,a',b,b',c,c'$ to denote the net flows into and out of this
    set. See Figure~\ref{fig:1} for definitions.  
    \begin{figure}
      \centering\includegraphics[scale=.8]{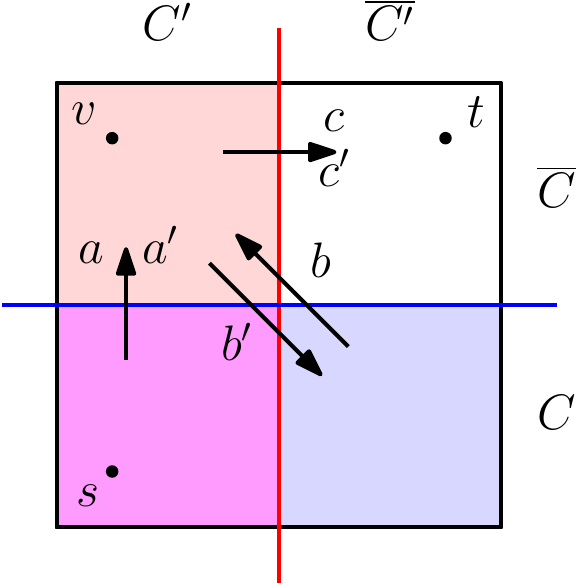}
      \caption{Diagram for the proof of Lemma~\ref{lem:1}.  Here the variables
        indicate the net flow between two sets under flows $f$ and $f'$.}
      \label{fig:1}
     \end{figure}
    By definition of $C$ and $C'$ there is no flow in $f$ from $\bar{C}$
    to $C$ nor is there flow in $f'$ from $\bar{C'}$ to $C'$ as otherwise
    vertices in the complementary cuts would be reachable from $s$.

    The flow conservation constraints require the flow into $\bar{C} \cap C'$ be
    matched by the outflow in both $f$ and $f'$.  This implies that $a + b = c$
    and $a' = b' + c'$.  In addition $a \ge a'$ and $c' \ge c$, because these
    edges must be at full capacity in $f$ and $f'$ respectively.  Combining
    these equalities and inequalities produces $a + b \le c'$ and $b' + c' \le
    a$.  Adding these two constraints together gives $a+b +b' + c' \le a + c'$.
    Since all values are non-negative we have $b = b' = 0$.  This implies $a =
    c$ and $a' = c'$.  Then, reusing $a \ge a'$ and $c' \ge c$ we conclude $a =
    c = a' = c'$.  This means that in $f'$ the edges going from $C \cap C'$ to
    $\bar{C} \cap C'$ are at full capacity, and thus no vertex in $\bar{C} \cap
    C'$ is reachable from $s$ in flow $f'$.  This is a contradiction.
  
    Since the flow from between $C \cap C'$ and $\bar{C} \cap C'$ is the same in
    both $f$ and $f'$, flow $f'$ witnesses that $C \cap C'$ is a minimum \cu of
    $G$.  This implies the ``moreover'' part of the statement and completes the
    proof.
\end{proof}

Note that this proof is similar to the ``lemma on a quadrangle'' from \cite{DKL76},
but that proof does not immediately go through because $\bar{C} \cap C'$ may
not be a \cu.

Lemma~\ref{lem:1} implies that the \FPC formula $\xi$ of Theorem~\ref{thm:cut}
defines a \emph{unique} \cu of the graph $G$ and for this reason we call it
the \emph{canonical} minimum \cu of $G$: $\MC{G,s,t}$.

%% file: min-odd-cuts.tex
The \emph{minimum odd cut} problem is closely related to the minimum \cu
problem.  Here the goal is to define a minimum odd cut of a graph $G$.  That is,
a cut of odd size whose value is minimum among all odd size cuts of $G$.  A
capacitated graph $G = (V,c)$ is \emph{symmetric} if for all $u,v \in V$,
$c(u,v) = c(v,u)$.  Observe that in symmetric graphs the set $C$ is a
minimum \cu iff its complement $\bar{C} \de V \bs C$ is a minimum $(t,s)$-cut.
In this section we prove that in each symmetric graph $G$ there is at least one pair of
vertices $s,t$ such that the canonical minimum \cu $\MC{G,s,t}$ is a minimum odd
cut of $G$.  The results and techniques discussed in this section are entirely
graph theoretic.

%This will then allow us to define the required separation oracle within \FPC.

Before continuing we must define several special types of cuts. We extend a
capacitated graph $G = (V,c)$ to a \emph{marked} capacitated
graph $G' = (V,c,M)$ with a marking $M \sse V$.  We call a vertex $v \in
V$ \emph{marked} if $v \in M$.  A cut $C$ of a marked graph $G$ is said to be
a \emph{marked cut}, if both $C$ and $\bar C$ contain a marked vertex.  A cut
$C$ of a graph $G = (V,c)$ with $|V|$ even is said to be an \emph{odd} cut if
$|C|$ is odd.  A marked cut $C$ of a marked graph $G = (V,c,M)$ with $|M|$ even
is said to be an \emph{odd} marked cut if $|C \cap M|$ is odd (note that this
corresponds to the simpler notion when $M = V$).  For any set $\mathcal{C}$ of cuts
we define the \emph{basic} cuts in $\mathcal{C}$ to be
$\condset{C \in \mathcal{C}}{\forall C' \in \mathcal{C}, C = C' \text{ or }
C \not\spe C'}$.  Note that if $\mathcal{C}$ is non-empty it must contain at
least one basic cut.  When $\mathcal{C}$ is the set of minimum \cus, the formula
$\xi$ of Theorem~\ref{thm:cut} defines the unique basic cut in $\mathcal{C}$.
We frequently describe sets of cuts by a sequence of the above adjectives and
determine meaning by first evaluating the adjective which appear closest to
the word ``cut''.  The most complex cuts we consider are ``basic minimum
odd marked cuts''.

Section~\ref{subsec:inter} develops several technical properties of cuts of
marked symmetric graphs.  Using these properties Section~\ref{subsec:existscut}
shows that there is a canonical minimum \cu of a symmetric graph $G$ which is
also a minimum odd cut of $G$.

\subsection{Intersections of Minimum Cuts}
\label{subsec:inter}

We now prove two technical properties involving the intersections of marked
cuts.  The first says that basic minimum marked cuts do not have complicated
intersections with basic minimum \cus.  The proof is similar in spirit to the
proof of Lemma~\ref{lem:1}

\begin{lem}
  \label{lem:3} Let $G = (V,c,M)$ be a marked symmetric graph.  Let
  $s, t \in M$ be distinct vertices and let $C$ be a basic minimum \cu of $G$.
  For every basic minimum marked cut $C'$ of $G$ one of the following holds: (i)
  $C \spe C'$, (ii) $C \cap C' = \es$ or (iii) $\set{s,t} \cap C' \neq \es$.
\end{lem}

\begin{proof}
   Fix any basic minimum marked cut $C'$ of $G$.  Suppose neither property
   (ii) or (iii) holds; it suffices to show that property (i) holds.  Thus our
   goal is to show that $C'' \de \bar{C} \cap C' = \es$ assuming that $C \cap
   C' \neq \es$ and $\set{s,t} \cap C' = \es$. See Figure~\ref{fig:2} for a
   diagram of the general configuration of these cuts and for the definitions of
   variables labelling the symmetric capacity crossing between the various sets.
 
   \begin{figure}
     \centering\includegraphics[scale=.8]{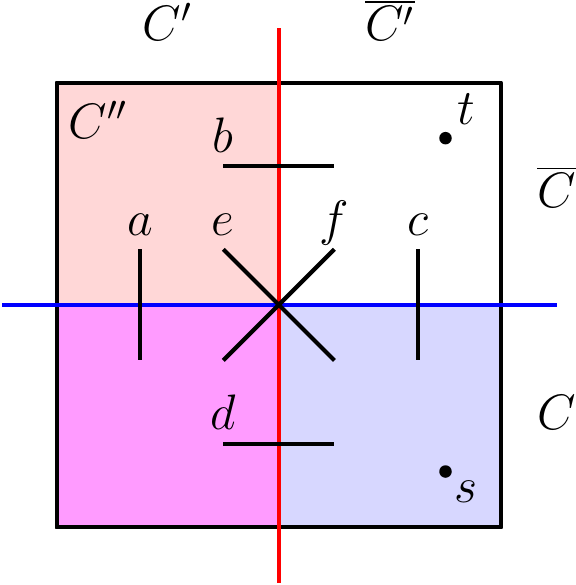}
     \caption{Diagram for the proof of Lemma~\ref{lem:3}.}
     \label{fig:2}
   \end{figure}

   Observe that $C \cap \bar{C'}$ is an \cu, however, it cannot be a minimum \cu
   because $C \supsetneq C \cap \bar{C'}$ (since $C \cap C'\neq \es$) and $C$ is a basic minimum \cu.  Thus
   $$c+d+e = \val(C \cap \bar{C'}) > \val(C) = a + c + e + f$$ and hence $d >
   a+f$.  

   Since $C'$ is a marked cut of $G$, $C'$ contains a marked vertex.  Suppose
   that $C''$ contains a marked vertex.  In this case $C''$ is marked cut of
   $G$, because it contains at least one marked vertex, but not all marked
   vertices (e.g., $s$).  Since $C'$ is a basic minimum marked cut of $G$ and
   $C'' \subsetneq C'$, $C''$ cannot be a minimum marked cut of $G$.  This
   implies that $$a+b+e = \val(C'') > \val(C') = b + d + e + f$$ and hence that
   $a > d+f$.  Combining this with the inequality $d > a + f$ derived from $C$
   being a basic minimum \cu we have $d > a + f > d + 2f$ which is a
   contradiction because all the edges have non-negative capacity.  Thus $C''$
   cannot contain a marked vertex.  This implies that $C \cap C'$ contains a
   marked vertex and is hence a marked cut of $G$.

   Suppose that $C'' \neq \es$.  Since $C'$ is a basic minimum marked cut of
   $G$, $C \cap C' \subsetneq C'$ cannot be a minimum marked cut of $G$.  Thus
   $$a+d+f = \val(C \cap C') > \val(C') = b+d+e+f$$ and hence $a > b+e$.
   Combining this with the value of $C$ and the fact that the edges are
   non-negative implies $$\val(C) = a + c + e + f > b + c + 2e + f \ge b + c + f
   = \val(\bar{C} \cap \bar{C'}).$$ Thus $\val(C)
   > \val(\bar{C} \cap \bar{C'})$.  As $\bar{C} \cap \bar{C'}$ contains $t$ but
   not $s$, it is an $(t,s)$-cut.  Moreover, $\bar{C} \cap \bar{C'}$ is an
   $(t,s)$-cut with value strictly less than that of $C$ which is a
   contradiction because $C$ is a minimum \cu and the capacities are symmetric.
   Therefore $C'' = \es$ and the proof is complete.
\end{proof}

The second lemma says that given a basic minimum odd marked cut $C$, there
exists a minimum marked cut $C'$ which does not have a complicated intersection
with $C$.  The proof is quite similar to those of
Lemmas~\ref{lem:1}~\&~\ref{lem:3}.

\begin{lem}
  \label{lem:4} Let $G = (V,c,M)$ be a marked symmetric graph with
  $|M|$ even.  Let $C$ be a basic minimum odd marked cut of $G$.  There
  exists a minimum marked cut $C'$ of $G$ such that one of the following holds:
  (i) $C \spe C'$ or (ii) $C \cap C' = \es$.
\end{lem}

\begin{proof}

  Fix any minimum marked cut $C'$ of $G$.  The complementary cut $\bar{C'}$ of
  $C'$ is also a minimum marked cut of $G$ because $G$ is symmetric.  If $C \spe C'$
  or $C \cap C' = \es$ the condition is immediately satisfied.  If $C' \spe C$,
  then $\bar{C'}$ has no intersection with $C$ and the minimum marked cut
  $\bar{C'}$ satisfies (ii).  If $\bar{C} \cap \bar{C'} = \es$, then
  $C \spe \bar{C'}$ and hence minimum marked cut $\bar{C'}$ satisfies (i).  In
  the case that none of these things happen we observe that $C \cap C'$,
  $C \cap \bar{C'}$, $\bar{C} \cap C'$ and $\bar{C} \cap \bar{C'}$ are all
  non-empty.  Because $C$ is an odd marked cut there are two disjoint (but
  symmetric) cases: 

  \case{1.}{$C \cap C'$ is an odd marked cut.}

     See Figure~\ref{fig:3} for a diagram of this case and for the definitions
     of variables labelling the edges crossing between the various sets.
     \begin{figure}
       \centering\includegraphics[scale=.8]{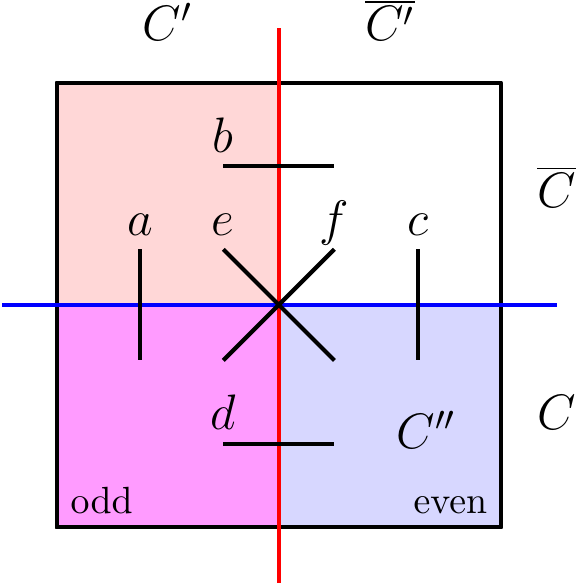}
       \caption{Diagram for the proof of Lemma~\ref{lem:4}.}
       \label{fig:3}
     \end{figure}
    The are two further subcases:

  \case{1.a.}{$\bar{C} \cap \bar{C'}$ contains a marked vertex.}  

      Since $C'$ is a marked cut of $G$, $C'$ contains a marked vertex.  Since
      $\bar{C} \cap \bar{C'}$ also contains an marked vertex,
      $\bar{C} \cap \bar{C'}$ is a marked cut of $G$.  Furthermore, $C'$ is a
      minimum marked cut of $G$ and hence
      $$b+d+e+f = \val(C') \le \val(\bar{C} \cap \bar{C'}) = b+c+f.$$ We resolve
      that $c \ge d+e.$ Similarly, since $C$ is basic, the odd cut $C \cap
      C' \subsetneq C$ of $G$ must have a larger value than $C$, and thus $d >
      c+e.$ Combining the two inequalities we conclude $d > c+e \ge d+2e.$ This
      is a contradiction because values are non-negative.

  \case{1.b.}{$\bar{C} \cap \bar{C'}$ contains no marked vertices.}  

      Since $\bar{C'}$ is a marked cut, it contains a marked vertex.  This
      implies that $C'' \de C \cap \bar{C'}$ contains a marked vertex, and hence
      that $C''$ is a marked cut of $G$.  Thus to satisfy (i) it suffices to
      show that $C''$ is a \emph{minimum} marked cut of $G$.

      %% \begin{claim}
      %%   $C''$ is a minimum marked cut of $G$.
      %% \end{claim}
      
      %% \begin{proof}
      
      Because $C$ is an odd marked cut and $C \cap C'$ is an odd marked cut,
      $C''$ contains an even number of marked vertices.  Since
      $\bar{C} \cap \bar{C'}$ contains no marked vertices, $\bar{C'}$ contains
      an even number of marked vertices.  This implies that $\bar{C} \cap C'$ is
      an odd marked cut of $G$.  Since $C$ is a minimum odd marked cut of $G$
      and $\bar{C} \cap C'$ is an odd marked cut we have that $$a+c+e+f
      = \val(C) \le \val(\bar{C} \cap C') = a + b + e.$$ This implies that
      $b \ge c + f.$ Similarly, since $C'$ is a minimum marked cut of $G$ and
      $C''$ is a marked cut of $G$ we have that $c \ge b + f.$ Combining these
      two inequalities $b \ge c + f \ge b + e + f.$ As all values are
      non-negative we must conclude that $b = c$ and $e = f = 0$.  This means
      that the value of $C''$ is $$\val(C'') = c+d+f = b+d+e+f = \val(C')$$ and
      we conclude that $C''$ is a minimum marked cut of $G$.
  
      %% \end{proof}

  \case{2.}{$C \cap \bar{C'}$ is an odd marked cut.}  

    Repeat Case 1 with $C'$ and $\bar{C'}$ swapped.  
\end{proof}

%% (Note: As a corollary we can strengthen the previous lemma to show it holds for
%% a canonical min cut $C'$ of the odd vertices $G$ since canonical min cuts are
%% subsets of all other min cuts.)

\subsection{Some Canonical Min $(s,t)$-Cuts are Min Odd Cuts}
\label{subsec:existscut}

In this subsection we show that for any symmetric graph $G = (V,c)$
there exists a pair of vertices $s,t \in V$ such that the canonical minimum \cu
$\MC{G,s,t}$ is a minimum odd cut of $G$.  

The intuition for the proof is as follows.  Let $G$ be a symmetric graph with
minimum odd cut $C$.  Mark all vertices in $G$.  We use Lemma~\ref{lem:4} to
locate a basic minimum marked cut $D$ of $G$ which $C$ does not partition 
(that is, $D$ is either contained in $C$ or disjoint from $C$).  If $D$ is an
odd cut, and because $D$ is basic, a canonical minimum \cu separating a vertex
$s \in D$ from a vertex $t \not\in D$ is also a minimum odd cut of $G$ and we
are done.  Otherwise $|D|$ is even and we form a new graph $G'$ by collapsing
$D$ into a new super-vertex $z$, and then setting the effected capacities so
that value of cuts which do not partition $D$ are unchanged.  Since $C$ does not
partition $D$ and $|D|$ is even, the collapsed version $C'$ of $C$ is a minimum
odd marked cut of $G'$.  We repeat this collapsing procedure maintaining a graph
$G'$ and minimum odd marked cut $C'$ until we locate a minimum marked cut $D$ of
$G'$ that is also minimum odd marked cut.  This provides marked vertices $s, t$
such that the canonical minimum \cu of $G'$ is a minimum odd cut of $G'$.  The
proof concludes using Lemma~\ref{lem:3} to translate this fact back to the
original graph $G$ and in doing so argues that the canonical minimum \cu of $G$
is a minimum odd cut.

% \marginnote{M: It this clear?  Perhaps a diagram would help.}

We now formalise the notion of collapsing a graph.  We begin by establishing
notation for substituting sets into sets.  Let $C, D \sse V$ and $z \not\in V$
such that either $C \spe D$ or $C \cap D = \es$. Define $C(z/D)$ to be a subset
of $V' \de (V \bs D) \cup \set{z}$ $$C(z/D) \de \begin{cases} (C \bs
D) \cup \set{z}, & C \spe D, \\ C, & C \cap D = \es, \end{cases}$$ and for
$C' \subseteq V'$ define $C'(D/z)$ to be a subset of
$V$ $$C'(D/z) \de \begin{cases} (C' \bs \set{z}) \cup D, & z \in C', \\ C', &
z \not\in C'.  \end{cases}$$ Observe that $(C(z/D))(D/z) = C$.

\algxio{\Coll}{fig:Collapse}{A subroutine to collapse a set of vertices in a
graph to a new single vertex.}{G,D,z}{A marked capacitated graph $G =
(V,c,M)$, a set $D \sse V$ and $z \not\in V$.\\}{The graph obtained from $G$ by
collapsing of the vertices in $D$ to $z$.}{
  \State $V' \gets V(z/D)$.
  \State $
         c'(u,w) \gets
         \begin{cases}
         c(u,w), &u,w \in V\bs D, \\
         \sum_{x \in D} c(u,x), &u \in V\bs D, w = z, \\
         \sum_{x \in D} c(x,w), &w \in V\bs D, u = z.
         \end{cases}
         $
  \State $M' \gets M \bs D$.
  \State \Return $(V',c',M')$.
}

The subroutine $\Coll(G,D,z)$ in Fig.~\ref{fig:Collapse} describes a
method of collapsing the vertex set $D \sse V$ in the marked graph $G = (V,c,M)$
to a single new super-vertex $z$.  This subroutine is designed to preserve a
basic minimum odd marked cut $C$ with respect to a well-chosen marked cut $D$.

\begin{lem}
  \label{lem:coll}
  Let $G = (V,c,M)$ be a marked symmetric graph with $|M|$ even. Let
  $z \not\in V$ and $C, D \sse V$ be marked cuts of $G$ such that $C \spe D$ or
  $C \cap D = \es$.    Define $C' \de C(z/D)$ and $G' \de
  (V',c',M') \de \Coll(G,D,z)$.
  \begin{enumerate}
  \item The value of $C$ in $G$ is identical to the value of $C'$ in $G'$.
  \item If $C$ is a basic minimum odd marked cut of $G$ and $|D \cap M|$ is
  even, then $C'$ is a basic minimum odd marked cut of $G'$, and $|M'|$ is even,
  non-zero and $M' \ssn M$.  
  \end{enumerate}
\end{lem}

\begin{proof}
  Suppose $C \spe D$.  The definition of \Coll implies that \alignedeq{\val(C)
  = \sum_{u \in C, v \in \bar{C}} c(u,v) &= \sum_{u \in C \bs D, v \in \bar{C}}
  c(u,v) + \sum_{u \in D, v \in \bar{C}} c(u,v) \\ &= \sum_{u \in C \bs D,
  v \in \bar{C}} c(u,v) + \sum_{v \in \bar{C}} c(z,v) \\ &= \sum_{u \in (C \bs
  D) \cup \set{z}, v \in \bar{C}} c(u,v) = \val(C').}  The case of $C \cap D
  = \es$ is analogous.  We conclude property 1 holds.

  Assume the hypothesis of property 2.  Observe that $D$ contains an even number
  of the marked vertices $M$ but not all of them because $D$ is a marked cut of
  $G$.  The subroutine \Coll sets $M' = M \bs D$.  Therefore $|M'|$ is even
  because $|D \cap M|$ and $|M|$ are even, and thus $M'$ meets the
  required conditions.

  We also determine that $C'$ is an odd marked cut of $G'$, because $|D \cap M|$
  is even and $C$ is an odd marked cut which either contains $D$ or is disjoint
  from $D$. 

  Suppose $C'$ is not a \emph{minimum} odd marked cut of $G'$.  Then there
  exists an odd marked cut $C''$ of $G'$ with smaller value than $C'$.  The cut
  $C''(D/z)$ is a marked cut of $G$ which does not partition $D$.  By property 1
  $\val(C''(D/z)) = \val(C'')$ and $\val(C) = \val(C(z/D)) = \val(C')$.  Hence
  $$\val(C''(D/z)) = \val(C'') < \val(C') = \val(C(z/D)) = \val(C).$$ However,
  $C''(D/z)$ is an odd marked cut because $C'$ is an odd marked cut, $z$ is not
  marked and $D$ contains an even number of marked vertices.  This means that
  $C$ is not a minimum odd marked cut of $G$ which contradicts the hypothesis.
  Therefore $C'$ is a \emph{minimum} odd marked cut of $G'$.

  Similarly, suppose $C'$ is not a \emph{basic} minimum odd marked cut of $G'$.
  Then there exists a minimum odd marked cut $C''$ of $G'$ with $C'' \ssn C'$.
  By property 1, it follows that $C''(D/z)$ is a minimum odd marked cut of $G$
  with $C''(D/z) \ssn C$, this contradicts the basicness of $C$.  Therefore $C'$
  is a \emph{basic} minimum odd marked cut of $G'$ and the proof is complete.

\end{proof}

With the key properties of \Coll established we are ready to prove the main 
result of this section.

\begin{thm}
  \label{thm:existscut}
  Let $G = (V,c,M)$ be a marked symmetric graph with even $|M| > 0$.  There exist
  $s,t \in M$ such that the canonical minimum \cu of $G$ is a minimum odd marked
  cut of
  $G$. 
\end{thm}
 
\begin{proof}
  Fix a basic minimum odd marked cut $C$ of $G$.  Figure~\ref{fig:witnessMC} describes
  the algorithm $\WMOC$ that computes vertices $s$ and $t$ witnessing the claim
  of the theorem from $C$ by iteratively collapsing $G$.  To prove the theorem
  it suffices to argue the algorithm halts and produces $(s,t) \in V^2$ such
  that the canonical minimum \cu $\MC{G,s,t}$ is a minimum odd marked cut.

\algxio{\WMOC}{fig:witnessMC}{An algorithm producing a witness for a
  minimum odd marked cut.}{G,C}{A symmetric graph $G = (V,c,M)$ and a minimum
    odd marked cut $C$ of $G$.\\}{$(s,t) \in M^2$ such that $\MC{G,s,t}$ is a minimum
    odd marked cut.}{ 
    \State $i \gets 0$.
    \State $G^0 \de (V^0,c^0,M^0) \gets (V,c,M)$.  
    \State $C^0 \gets C$.
    \While{true} \label{wmoc:guard}
      \State Let $D^i$ be a basic minimum marked cut of $G^i$ such that
    $C^i \spe D^i$ or $C^i \cap D^i = \es$. 
      \If{$D^i$ is an odd marked cut of $G^i$} \label{wmoc:test}
      \State \Return $(s,t)$ with $s \in D^i \cap M^i$ and $t \in M^i \bs
        D^i$. \label{wmoc:output} \EndIf
      \State $G^{i+1} \gets$ $\Coll(G^i,D^i,z^i)$.
      \State $C^{i+1} \gets C^i(z^i/D^i)$.
      \State $i \gets i+1$.   
    \EndWhile
  }

  As the algorithm runs it maintains the invariant that $C^i$ is a basic minimum
  odd marked cut of $G^i$.  Observe that this is initially true for $C^0 = C$
  because $M^0 = M$ and $C$ is a basic minimum odd marked cut of $G$.  Suppose $C^i$ is
  a basic minimum odd marked cut of $G^i$.  The basic minimum marked cut $D^i$
  of $G^i$ with $C^i \spe D^i$ or $C^i \cap D^i = \es$ is guaranteed to exist
  by Lemma~\ref{lem:4} (if the cut given by that lemma is not basic
  there must be a basic minimum marked cut strictly within it that continues to
  satisfy the intersection properties with $C^i$).  If $D^i$ is an odd marked
  cut, the algorithm halts at line~\ref{wmoc:output}.  Otherwise the graph $G^i$
  and cut $C^i$ are collapsed relative to $D^i$.  The second property of
  Lemma~\ref{lem:coll} implies that $C^{i+1}$ is a basic minimum odd marked cut
  of $G^{i+1}$.  Thus the invariant holds. 

  Lemma~\ref{lem:coll} also implies that as the algorithm runs, $|M^i|$ is even
  and $M^{i+1} \ssn M^i$.  The invariant and $M^0 = V$ imply that the test in
  line~\ref{wmoc:test} will be successful and cause the algorithm to halt within
  $\frac{|V|}{2}$ iterations.  Because $D^i$ is a marked cut of $G^i$, when
  line~\ref{wmoc:output} is reached $D^i \cap M^i$ and $M^i\bs D^i$ are non-empty
  disjoint sets.  This means that distinct $s$ and $t$ exist and are returned by
  the algorithm.  Let $r$ be the value of $i$ when the algorithm halts.  Fix any
  $s \in D^r \cap M^r$ and $t \in D^r \bs M^r$.  We use the shorthand $K^i$ to
  denote the canonical minimum \cu $\MC{G^i,s,t}$ for $0 \le i \le r$.  It
  remains to argue that $K^0 = \MC{G,s,t}$ is a minimum odd marked cut.

  Since the cut $D^r$ is a minimum marked cut of $G^r$ and $s,t \in M^r$, $D^r$
  is also a minimum \cu of $G^r$ and it has the same value as the canonical
  minimum \cu $K^r$.  This implies that $K^r$ is a minimum
  marked cut of $G^r$ because $s,t \in M^r$.  By Lemma~\ref{lem:1}, $D^r$
  contains $K^r$, but $D^r$ is also basic, so we conclude that $D^r
  = K^r$.

  The first property of Lemma~\ref{lem:coll} implies that $\val(C^i(z^i/D^i))
  = \val(C^{i+1})$ for all $0 \le i < r$, and hence that $\val(C^i) = \val(C)$
  for all $0 \le i \le r$.  Since $K^r$ is a minimum marked cut of
  $G^r$ and $C^r$ is a marked cut of $G^r$, $\val(K^r) \le \val(C^r)
  = \val(C)$.  The first part of Lemma~\ref{lem:coll} also implies that
  $\val(K^{i+1}(D^i/z^i)) = \val(K^{i+1})$ for all $0 \le i <
  r$.  Since $K^{i+1}(D^i/z^i)$ is an \cu of $G^i$,
  $\val(K^i) \le \val(K^{i+1}(D^i/z^i))$ for all $0 \le i <
  r$.  Hence we conclude that $\val(K^0) \le \val(K^r) \le
  \val(C)$.

  It remains to argue that $K^0$ is an odd marked cut.  Since $K^r$ is an odd
  marked cut, it suffices to show that $K^i$ is an odd marked cut if $K^{i+1}$
  is an odd marked cut, for all $0 \le i < r$.  To this end assume that
  $K^{i+1}$ is an odd marked cut.  Apply Lemma~\ref{lem:3} with $G^i$, $K^i$ and
  $D^i$; we note that (i) $K^i \spe D^i$, (ii) $K^i \cap D^i = \es$, or (iii)
  $\set{s,t} \cap D^i \neq \es$.  Property (iii) cannot hold because $s$ and $t$
  are selected after $D^i$ was collapsed.  This means that $K^i$ either contains
  all of $D^i$ or is disjoint from $D^i$.  Because the $K^i$ and $K^{i+1}$ are
  canonical \cus, $$K^i \sse K^{i+1}(D^i/z^i) \sse (K^i(z^i/D^i))(D^i/z^i) =
  K^i.$$ As the algorithm did not halt at step $i$, $|D^i \cap M^i|$ is even and
  thus $K^i$ is an odd marked \cu.

  We conclude that $K^0$ is an odd marked cut with value at most that of a
  minimum odd marked cut $C$ of $G$.  Therefore $K^0 = \MC{G,s,t}$ is a minimum odd
  marked cut.
\end{proof}

Consider the following procedure for locating a set of minimum odd marked cuts
in a marked symmetric graph $G = (V,c,M)$: For all distinct $s, t \in M$ compute
the canonical minimum \cu $\MC{G,s,t}$, eliminate those cuts which are not odd,
then eliminate those cuts which are not minimal.  Theorem~\ref{thm:existscut}
indicates that some cuts remain and that those cuts are minimum odd cuts of $G$.
Note that the algorithm \WMOC in the proof of Theorem~\ref{thm:existscut} is
used only in the analysis and not actually run during the above procedure.  This
simple procedure for defining a non-empty set of minimum odd cuts is critical to
expressing the separation problem for the matching polytope in \FPC.

%% file: matching.tex
Let $G = (V,E)$ be an undirected graph.  A \emph{matching} $M \subseteq E$ is
defined by the property that no two edges in $M$ are incident to the same
vertex.  A matching $M$ is \emph{maximum} if no matchings with size larger than
$M$ exist.  A maximum matching is \emph{perfect} if every vertex in $G$ is
incident to some edge in the matching (i.e., $|M| = \frac{|V|}{2}$).

\subsection{Maximum Matching Program}

Maximum matching has an elegant representation as a linear program.  In fact, it
is an instance of a slightly more general problem: $b$-matching.  Let $c \in
\nnQQ^E$, $b \in \NN^V$ and $A \in \set{0,1}^{V \times E}$ be the incidence
matrix of the undirected graph $G = (V,E)$: the columns of $A$ correspond to the
edges $E$ and the rows to the vertices $V$, and $A_{ve}=1$ if edge $e$ is
incident on vertex $v$.  Alternatively we view edges $e \in E$ as two-element
subsets of $V$.  The goal of the $b$-matching problem is to determine an optimum
of the following \emph{integer} linear program
\begin{equation}
  \label{eqn:matchingLP0}
    \max\:  c^\top y \quad \text{ subject to} \quad  Ay \le b, \:
     y \ge 0^E. 
\end{equation}
We obtain the usual maximum matching problem in the special case where $b = 1^V$
and $c = 1^E$.

Generically, integer programming is \NP-complete, so instead of trying to
directly solve the above program we consider the following relaxation as a
rational   
linear program.
\begin{equation}
  \label{eqn:matchingLP}
  \begin{aligned}
    \max\;& c^\top y \quad\quad \text{ subject to} \\ 
    & Ay \le b, \\
    & y \ge 0^E, \\
    & y(W) \le \frac{1}{2}(b(W) - 1),\;\; \forall W \subseteq V \text{ with }
    b(W) \text{ odd},
  \end{aligned}
\end{equation}
where $y(W) \de \sum_{e \in E, e \sse W} y_e$ and $b(W) \de \sum_{v \in W} b_v.$
Here we have added a new set of constraints over subsets of the vertices.  The
integral points which satisfy \eqref{eqn:matchingLP0}, also satisfy the
additional constraints that are added in \eqref{eqn:matchingLP}.  To see this,
let $y$ be a feasible integral solution, consider some set $W$ with $b(W)$ odd.
If $|W| = 1$, then $y(W) = 0$ because no edges have both endpoints in $W$, so
assume $|W| \ge 2$.  It follows that $2y(W) \le b(W)$, by summing the constraints
of $Ay \le b$ over $W$ with respect to only the edges with \emph{both} endpoints
in $W$.  Since $b(W)$ is odd, $\frac{1}{2}b(W)$ is half integral, but $y(W)$ is
integral because $y$ is an integral solution; this means the constraint $y(W)
\le \frac{1}{2}(b(W) - 1)$ is a valid constraint for all integral solutions.
In fact \cite{E65} shows something stronger.
% the extremal points of the linear program \eqref{eqn:matchingLP} are exactly the maximum $b$-matchings of $A$.

\begin{lem}[{\cite[Theorem P]{E65}}]
  \label{lem:int}
  The extremal points of the linear program \eqref{eqn:matchingLP} are integral
  and are the extremal solutions to the $b$-matching problem.
\end{lem}

Thus to solve $b$-matching it suffices to solve the relaxed linear program
\eqref{eqn:matchingLP}.  As mentioned before, it will not be possible to show
that \FPC can generally define a particular maximum matching, there can be
simply too many.  However, the above lemma means that the existence of a (likely
non-integral) feasible point $y$ of \eqref{eqn:matchingLP} with value $c^\top y$
witnesses the existence of a maximum $b$-matching with value at least $c^\top
y$.  In addition, the number of constraints in this linear program is
exponential in the size of the graph $G$.  Thus, we cannot hope to interpret
this linear program directly in $G$, using \FPC.  Rather what we can show is
that there is an \FPC interpretation which, given $G$, $b$ and $c$, expresses
the separation problem for the $b$-matching polytope in the linear
program~\eqref{eqn:matchingLP}.  Combining this with
Theorem~\ref{thm:opt-to-sep} gives an \FPC interpretation expressing the
$b$-matching optimum.

\subsection{Expressing Maximum Matching in \FPC}

%% We assume that instances of $b$-matching are given by a matrix $A \in
%% \set{0,1}^{V \times E}$, a bound vector $b \in \NN^V$ and an objective vector
%% $c \in \nnQQ^E$ for a vertex set $V$ and edge set $E$.\
%footnote{We
%  could instead take the matrix in the more standard vertex-vertex adjacency
%  form.  This would have complicated the discussion.   It is an easy exercise
%  to show that the two versions are equivalent in \FPC.}  
The $b$-matching
polytopes have a natural representation over $\vocmatch \de \vocmat \uplus
\vocvec$.  Although the number of constraints in the $b$-matching
  polytope may be large, the
individual constraints have size at most a polynomial in the size of the
matching instance.  Thus this representation is well-described.

We now describe an \FPC interpretation expressing the separation problem for the
$b$-matching polytope given a $\vocmatch$-structure coding the matrix $A$ and
bound vector $b$. As in the explicit
constraint setting, our approach is to come up with a definable set of violated
constraints iff the candidate point is infeasible.  We then define a canonical
violated constraint by summing this definable violated set.  Identifying violated
vertex and edge constraints can easily be done in \FPC as before.  However, it is
not immediately clear how to do this for the odd set constraints.

To overcome this hurdle we follow the approach of \cite{PR82}.  Let $y$ be
point which we wish to separate from the matching polytope.  Define $s \de b -
Ay$ to be the slack in the constraints $Ay \le b$.  Analogous to $b(W)$, define
$s(W) \de \sum_{v \in W} s_v$.  Observe that $2y(W) + y(W : V\bs W) + s(W) =
b(W)$ (here $y(W : V \bs W)$ is sum of edge variables with one endpoint in $W$
and one in $V \bs W$).  This translates the constraints $y(W) \le
\frac{1}{2}(b(W)-1)$ exactly to $y(W : V \bs W) + s(W) \ge 1$.  This means to
find a violated constraint of this type it suffices to find $W$ such that $y(W :
V \bs W) + s(W) < 1$.

Define a marked symmetric graph $H$ over vertex set $U \de V \cup \set{z}$
where $z$ is a new vertex.  Let $H$ have symmetric capacity $d$: $d(u,v) \de
y_e$ when $u,v \in V$ and $u,v \in e$, and $d(u,v) \de s_v$ when $u = z$ and
$v \in V$.
%% $$
%% c(u,v) \de
%% \begin{cases}
%%   y_e & e \in E \text{ with } u,v \in e, \\
%%   s_v, & v \in V, u = z.
%% \end{cases}
%% $$ 
Let $M \de \condset{v \in V}{b_v \text{ is odd}}$.  If $|M|$ is odd, add $z$ to $M$.
Thus we have a marked symmetric graph $H = (U,d,M)$.  Consider any odd marked cut $W$
of $H$, without loss of generality $z \not\in W$ (otherwise, take the
complement).  Observe that the value of edges crossing the cut is exactly $y(W :
V \bs W) + s(W)$; also note that $s(W)$ is odd.  Thus there is a minimum
odd marked cut $W$ of $H$ with value less than $1$ iff there is a violated odd set
constraint in \eqref{eqn:matchingLP}.  

By Theorem~\ref{thm:existscut}, there is a violated odd set constraint iff for
some $s,t \in M$ the canonical minimum \cu is an minimum odd marked cut with
value less than $1$.  We conclude, using Theorem~\ref{thm:cut} and
Lemma~\ref{lem:1}, that we can define a family of violated set constraints
within \FPC.  Summing these defined violated constraints produces a canonical
violated constraint which must be non-trivial by
Proposition~\ref{prop:sum-constraint}.  Thus, as in
Theorem~\ref{thm:explicit-sep} there is an \FPC interpretation expressing the
separation problem for the polytope in the linear program
\eqref{eqn:matchingLP}.

\begin{lem}
  \label{lem:sep-matching}
  There is an \FPC interpretation of $\vocvec$ in $\vocmatch \uplus \vocvec$
  expressing the separation problem for the $b$-matching polytopes
  with respect to  their natural representation as $\vocmatch$-structures.
\end{lem}

Like the maximum flow problem in Section~\ref{sec:maxflow}, the $b$-matching polytope
is both compact and nonempty.  
By combining Lemma~\ref{lem:sep-matching} and Theorem~\ref{thm:opt-to-sep} with
respect to the natural well-described representation of $b$-matching polytopes,
we conclude that there is an \FPC interpretation expressing the value of the
maximum $b$-matching of a graph.

\begin{thm}
  There is an \FPC interpretation of $\vocnum$ in $\vocmatch \uplus \vocvec$
  which takes a $\vocmatch \uplus \vocvec$-structure coding a $b$-matching
  polytope $P$ and a vector $c$ to a rational number $m$ indicating the value of
  the maximum $b$-matching of $P$ with respect to $c$.
\end{thm}

%% \begin{thm}
%%   Let $A$ be a relation interpreted as the vertex-edge adjacency matrix of an
%%   undirected graph.  Let $b$ be a relation interpreted as associating positive
%%   integers with the vertices of the graph.  Let $c$ be a relation interpreted as
%%   associating non-negative real numbers to the edges of the graph.  There is a
%%   \FPC formula on $\set{A,b,c}$-structures defining the value of the maximum
%%   $b$-matching of graphs.
%% \end{thm}

%% file: conc.tex
Our main result is that the linear programming problem can be expressed
in fixed-point logic with counting---indeed, that the linear optimisation problem can be expressed in \FPC
for any class of polytopes for which the separation
problem can be defined in \FPC.  As a consequence, we solve an open problem of \cite{BGS99} concluding
that there is a formula of fixed-point logic with counting which defines the size of a maximum $b$-matching in a graph.  This
is one demonstration of the power of the ellipsoid method and linear
optimisation that can be brought to bear even in the setting
of logical definability.  From here, there are number of natural research directions to
consider.

\paragraph{Convex programming}  A polytope is an instance of much more %
general geometric object: a convex set.  The robust
nature of the ellipsoid method means it has
been extended to help solve more general optimisation
problems, e.g., semi-definite programs and quadratic programs.  It
seems likely that our methods can be extended to these settings.

%% \marginnote{M: I haven't thought about this at length with respect to the new
%% proof in Section~\ref{sec:opt-sep}.  It seems quite likely that the proof should
%% extend to general convex set because the main property used is
%% convexity.
%% A: I have modified the last sentence to reflect this.}

%% difficult optimization problems.  For example, the standard method to solve the
%% separation problem for semi-definite programs involves locating eigenvectors
%% with nonpositive eigenvalues (indeed, perhaps an eigenbasis) to witness
%% violations of positive-semi-definitiness constraints.  Such choices cannot be
%% expressed in the logic \FPC. 

 %% Note that the proof of Theorem~\ref{thm:opt-to-sep}
%% skirts the choice of basis issue by working simultaneously with a set indistinguishable vectors forming part of a basis. % derived from the matrix defining the ellipsoid and indexed by the variable set.

\paragraph{Completeness}  Linear programming is complete for polynomial time
under logspace reductions \cite{DLR79}.  It follows from our results
that it cannot be complete for \PT under logical reductions such as
first-order interpretations, since this would imply that \PT is
contained in \FPC.  Could it still be the case that linear
programming is complete for \FPC under such weak reductions?  Or
perhaps \FOC reductions?
Even if linear programming is not complete, 
there may be other interesting combinatorial problems that
can be expressed in \FPC via reduction to linear programming. 
%(see \cite{KV12} for some possibilities).  
There has also been 
%considerable 
some work
examining generalisations and improvements to the $b$-matching approach we
followed 
%\cite{PR82} 
(e.g., \cite{CF96%,LRT04
}), and 
it is possible these results
%it is worth examining how much of this
can also be replicated in \FPC.

\paragraph{LP hierarchies and integrality gaps}  Another intriguing connection
between counting logics and linear programming is established
in~\cite{AM12,GO12} where it is shown that the hierarchy of
Sherali-Adams relaxations \cite{SA90} 
of the graph isomorphism integer
program interleaves with equivalence in $k$-variable logic with
counting ($C^k)$.  It is suggested~\cite{AM12} that inexpressibility results for
$C^k$ could be used to derive integrality gaps for such relaxations.
It is a consequence of the results in this paper
that the Sherali-Adams approximations of not only isomorphism, but of
other combinatorial problems can be expressed in \FPC.  Do our results provide another
route to using inexpressibility in \FPC to prove
integrality gaps?

%% In particular, these study
%% the tight interleaving of the successive approximations to graph isomorphism by
%% the Sherali-Adams hierarchies \cite{SA90} of linear programs and pebble-game
%% equivalence with counting. \emph{BH: Mention here \FPC-definability of $C^k$-equivalence?}\marginnote{Ref 2}

%% file: ChoicelessLP.bbl
\providecommand{\bysame}{\leavevmode\hbox to3em{\hrulefill}\thinspace}
\providecommand{\MR}{\relax\ifhmode\unskip\space\fi MR }
% \MRhref is called by the amsart/book/proc definition of \MR.
\providecommand{\MRhref}[2]{%
  \href{http://www.ams.org/mathscinet-getitem?mr=#1}{#2}
}
\providecommand{\href}[2]{#2}
\begin{thebibliography}{DGHL09}

\bibitem[ABD09]{ABD09}
A.~Atserias, A.~Bulatov, and A.~Dawar, \emph{Affine systems of equations and
  counting infinitary logic}, Theor. Comput. Sci. \textbf{410} (2009), no.~18,
  1666--1683.

\bibitem[AM12]{AM12}
A.~Atserias and E.~Maneva, \emph{{S}herali-{A}dams relaxations and
  indistinguishability in counting logics}, ITCS, ACM, 2012, pp.~367--379.

\bibitem[BG05]{Bla05}
A.~Blass and Y.~Gurevich, \emph{A quick update on open problems in
  {B}lass-{G}urevich-{S}helah's article `{O}n polynomial time computations over
  unordered structures'}, Online at
  \url{http://research.microsoft.com/~gurevich/annotated.html}, 2005, [Accessed
  July 19, 2010].

\bibitem[BGS99]{BGS99}
A.~Blass, Y.~Gurevich, and S.~Shelah, \emph{Choiceless polynomial time}, Ann.
  Pure Appl. Logic \textbf{100} (1999), 141--187.

\bibitem[BGS02]{BGS02}
\bysame, \emph{On polynomial time computation over unordered structures}, J.
  Symbolic Logic (2002), 1093--1125.

\bibitem[CF96]{CF96}
A.~Caprara and M.~Fischetti, \emph{$\{$0, 1/2$\}$-{C}hv{\'a}tal-{G}omory cuts},
  Math. Program. \textbf{74} (1996), no.~3, 221--235.

\bibitem[CFI92]{CFI92}
J-Y. Cai, M.~F\"{u}rer, and N.~Immerman, \emph{An optimal lower bound on the
  number of variables for graph identification}, Combinatorica \textbf{12}
  (1992), no.~4, 389--410.

\bibitem[CH82]{CH82}
A.~Chandra and D.~Harel, \emph{Structure and complexity of relational queries},
  J. Comput. Syst. Sci. \textbf{25} (1982), no.~1, 99--128.

\bibitem[Dan63]{D98}
G.~Dantzig, \emph{Linear programming and extensions}, Princeton University
  Press, 1963, (most recent edition published in 1998).

\bibitem[DGHL09]{DGHL09}
A.~Dawar, M.~Grohe, B.~Holm, and B.~Laubner, \emph{Logics with rank operators},
  LICS, IEEE, 2009, pp.~113--122.

\bibitem[DKL76]{DKL76}
E.A. Dinitz, A.V. Karazanov, and M.V. Lomonosov, \emph{On the structure of the
  system of minimum edge cuts in a graph}, Studies in Discrete Optimizations
  (1976), pp. 290--306 (Russian).

\bibitem[DLR79]{DLR79}
D.~Dobkin, R.J. Lipton, and S.~Reiss, \emph{Linear programming is log-space
  hard for \emph{P}}, Inform. Process. Lett. \textbf{8} (1979), no.~2, 96--97.

\bibitem[Edm65]{E65}
J.~Edmonds, \emph{Maximum matching and a polyhedron with $0,1$ vertices}, J.
  Res. Nat. Bur. Stand. \textbf{69~B} (1965), 125--130.

\bibitem[EF99]{Ebb99}
H.D. Ebbinghaus and J.~Flum, \emph{Finite model theory}, Springer, 1999.

\bibitem[GLS81]{GLS81}
M.~Gr\"{o}tschel, L.~Lov\'{a}sz, and A.~Schrijver, \emph{The ellipsoid method
  and its consequences in combinatorial optimization}, Combinatorica \textbf{1}
  (1981), 169--197.

\bibitem[GLS88]{GLS88}
M.~Gr{\"o}tschel, L.~Lov{\'a}sz, and A.~Schrijver, \emph{Geometric algorithms
  and combinatorial optimization}, Springer-Verlag Berlin / New York, 1988.

\bibitem[GO12]{GO12}
M.~Grohe and M.~Otto, \emph{{Pebble Games and Linear Equations}}, CSL, Schloss
  Dagstuhl--Leibniz-Zentrum fuer Informatik, 2012, pp.~289--304.

\bibitem[Gro10]{Gro10}
M.~Grohe, \emph{Fixed-point definability and polynomial time on graph with
  excluded minors}, LICS, IEEE, 2010, pp.~179--188.

\bibitem[Hol10]{H10}
B.~Holm, \emph{Descriptive complexity of linear algebra}, Ph.D. thesis,
  University of Cambridge, 2010.

\bibitem[Imm86]{I86}
N.~Immerman, \emph{Relational queries computable in polynomial time}, Inform.
  Control \textbf{68} (1986), no.~1-3, 86--104.

\bibitem[Imm99]{Imm99}
\bysame, \emph{Descriptive complexity}, Springer-Verlag, 1999.

\bibitem[Kha79]{K79}
L.G. Khachiyan, \emph{{A polynomial algorithm in linear programming.}}, Dokl.
  Akad. Nauk SSSR \textbf{244} (1979), 1093--1096 (Russian).

\bibitem[Kha80]{K80}
\bysame, \emph{Polynomial algorithms in linear programming}, USSR Comp. Math.
  Math \textbf{20} (1980), no.~1, 53--72.

\bibitem[Lib04]{Libkin}
L.~Libkin, \emph{Elements of finite model theory}, Springer, 2004.

\bibitem[PR82]{PR82}
M.W. Padberg and M.R. Rao, \emph{Odd minimum cut-sets and b-matchings}, Math.
  Oper. Res. \textbf{7} (1982), no.~1, 67--80.

\bibitem[Ros10]{Ros10}
B.~Rossman, \emph{Choiceless computation and symmetry}, Fields of Logic and
  Computation, Springer, 2010, pp.~565--580.

\bibitem[SA90]{SA90}
H.D. Sherali and W.P. Adams, \emph{A hierarchy of relaxations between the
  continuous and convex hull representations for zero-one programming
  problems}, SIAM Journal on Discrete Mathematics \textbf{3} (1990), no.~3,
  411--430.

\bibitem[Sho72]{S72}
N.Z. Shor, \emph{Utilization of the operation of space dilatation in the
  minimization of convex functions}, Cybern. Syst. Anal. \textbf{6} (1972),
  no.~1, 7--15.

\bibitem[Sho77]{S77}
\bysame, \emph{Cut-off method with space extension in convex programming
  problems}, Cybern. Syst. Anal. \textbf{13} (1977), no.~1, 94--96.

\bibitem[ST04]{ST04}
D.A. Spielman and S.-H. Teng, \emph{Smoothed analysis of algorithms: Why the
  simplex algorithm usually takes polynomial time}, JACM \textbf{51} (2004),
  no.~3, 385--463.

\bibitem[Var82]{V82}
M.~Vardi, \emph{The complexity of relational query languages}, STOC, ACM, 1982,
  pp.~137--146.

\bibitem[YN76]{Y76}
D.B. Yudin and A.S. Nemirovskii, \emph{Informational complexity and efficient
  methods for the solution of convex extremal problems}, Matekon \textbf{13}
  (1976), no.~2, 3--25.

\end{thebibliography}
